\documentclass[a4paper,onecolumn,11pt,unpublished]{quantumarticle}
\pdfoutput=1
\usepackage[utf8]{inputenc}
\usepackage[english]{babel}
\usepackage[T1]{fontenc}
\usepackage{amsmath,amssymb,amsthm,mathtools}
\usepackage[margin=1in]{geometry}
\usepackage{microtype}
\usepackage{booktabs,longtable,array}
\usepackage{enumitem}
\usepackage[numbers,sort&compress]{natbib}
\usepackage[colorlinks=true,allcolors=blue]{hyperref}
\usepackage{xurl}
\usepackage{bm}

\newtheorem{theorem}{Theorem}
\newtheorem{lemma}{Lemma}
\newtheorem{proposition}{Proposition}
\newtheorem{corollary}{Corollary}
\theoremstyle{definition}
\newtheorem{definition}{Definition}

\newcommand{\R}{\mathbb R}
\newcommand{\C}{\mathbb C}
\newcommand{\N}{\mathbb N}
\newcommand{\Sep}{\operatorname{Sep}}
\newcommand{\Bisep}{\operatorname{Bisep}}
\newcommand{\Tr}{\operatorname{Tr}}
\newcommand{\conv}{\operatorname{conv}}
\newcommand{\Hess}{\nabla^2}
\newcommand{\Realif}{\mathcal R}
\newcommand{\ket}[1]{|#1\rangle}
\newcommand{\bra}[1]{\langle#1|}
\newcommand{\proj}[1]{\ket{#1}\bra{#1}}
\newcommand{\Ran}{\operatorname{Ran}}

\DeclareMathOperator{\re}{Re}
\DeclareMathOperator{\im}{Im}

\title{Separability rigidity of Gaussian states under positive domination}

\author{Evgeny Shchukin}
\address{Johannes-Gutenberg University of Mainz, Institute of Physics, Mainz, Germany}
\email{shchukin@uni-mainz.de}

\author{Peter van Loock}
\address{Johannes-Gutenberg University of Mainz, Institute of Physics, Mainz, Germany}
\email{loock@uni-mainz.de}

\date{}
\begin{document}

\begin{abstract}
We prove a separability rigidity theorem for finite-mode Gaussian states. Let
$\rho_G$ be Gaussian and let $\mathcal I$ be a fixed partition of the modes. If
$\rho_G$ dominates a nonzero positive operator $T$ that is separable across
$\mathcal I$, $0\neq T\preccurlyeq \rho_G$, then $\rho_G$ itself is separable
across $\mathcal I$. The dominated operator is arbitrary and need not be
Gaussian. The theorem has three applications. First, any Gaussian channel
admitting a nonzero entanglement-breaking completely positive subchannel is
itself entanglement breaking, even when that subchannel is non-Gaussian. Second,
if a resource state dominates a positive multiple of a faithful separable state,
no finite-copy, postselected separability-preserving protocol can prepare an
entangled Gaussian state exactly with nonzero success probability. Third, a
connected graph of entangled Gaussian two-mode marginals certifies genuine
multipartite entanglement of every compatible global state, even if the global
state is non-Gaussian; the certification is robust to sufficiently small
trace-norm errors in the marginals. As a structural corollary, mixing states
separable across different partitions cannot produce new Gaussian states, so
full inseparability and genuine multipartite entanglement coincide for
finite-mode Gaussian states and identical copies of a biseparable Gaussian state
cannot activate genuine multipartite entanglement. 
\end{abstract}

\maketitle

\section{Introduction}

Entanglement between every partition of a collection of modes into two groups is
called \emph{full inseparability}. A stronger notion, \emph{genuine multipartite
entanglement} (GME), rules out mixtures of states separable across possibly
different divisions. For general mixed states these notions differ: the
separable division can change from one component of a mixture to another.

Leskovjanov\'a et al.\ conjectured that this distinction disappears for Gaussian
states \cite{Leskovjanova2026}. Their numerical investigation concerns Gaussian
states whose covariances also occur in non-Gaussian biseparable states. The
distinction between a state and its covariance is essential: a covariance
biseparability condition is not sufficient even under a Gaussian promise
\cite{Baksova2025,Leskovjanova2026}.

We prove the conjectured equivalence by establishing a stronger statement. If
even a nonzero portion of a Gaussian state is separable across a specified
partition, in the precise positive-operator sense below, the whole Gaussian
state is separable across that partition. This statement concerns arbitrary
separable portions, not only Gaussian ones. The familiar Gaussian covariance
criterion and its construction by classical random displacements are established
results \cite{WernerWolf2001}; an operational bipartite criterion was
developed in \cite{Giedke2001}. Our argument derives the required product
covariance from an arbitrary dominated separable operator.

There are two infinite-dimensional issues even with finitely many modes. First,
separable operators need not admit countable pure-product decompositions
\cite{Holevo2005}. We use finite-mixture approximation and trace-norm continuity
of Husimi derivatives instead. Second, biseparability is a \emph{closed} convex
hull. Positivity and uniform trace tails supply the compactness needed to
recover an exact finite decomposition from that closure.

\section{Setting and main results}\label{sec:setting}

Consider $n$-mode quantum states. A $k$-partition $\mathcal{I}$ of the modes is
any grouping of these modes into $k$ non-empty disjoint sets
\begin{equation}
      \mathcal{I} = \{I_1,\ldots,I_k\}, \quad
      I_1 \cup \ldots \cup I_k = [n].
\end{equation}
We define the space $\Sep(\mathcal{I})$ of $\mathcal{I}$-separable states as
follows:
\begin{definition}\label{def:sep} For a fixed partition
$\mathcal{I}=\{I_1,\ldots,I_k\}$ of $n$ modes, let
\begin{equation}\label{eq:SepI}
 \Sep(\mathcal{I})=
 \overline{\conv}
 \{\hat{\varrho}_1 \otimes \ldots \otimes \hat{\varrho}_k:
 \ \hat{\varrho}_j\ \text{is a state on modes}\ I_j\},
 \end{equation}
 where the closure is taken with respect to the trace norm
 $\|T\|_1=\Tr\sqrt{T^\dagger T}$. For $n\geqslant 2$, let $\mathfrak{I}$ be an
 arbitrary nonempty collection of arbitrary partitions, not necessarily with the
 same number of parts. Define the set of all $\mathfrak{I}$-separable states as
 \begin{equation}\label{eq:Sep-I}
    \Sep(\mathfrak{I}) = \conv\biggl(\bigcup_{\mathcal{I}\in\mathfrak{I}}\Sep(\mathcal{I})\biggr).
 \end{equation}
Let $\Pi_n$ be the set of nontrivial bipartitions, identifying a cut with its
complement. Then the set of biseparable states is defined as
\begin{equation}\label{eq:Bisep}
 \Bisep_n = \Sep(\Pi_n).
\end{equation}
A state is fully inseparable if it lies in no $\Sep(\pi)$ for all $\pi \in
\Pi_n$, and is GME if it lies outside $\Bisep_n$. The outer trace-norm closure
in Eq.~\eqref{eq:Sep-I} can in fact be omitted: because the number of partitions
is finite, the algebraic convex hull of the fixed-cut separable sets is already
trace-norm closed. This is proved in Appendix~\ref{app:D}.
\end{definition}

A nonzero positive trace-class operator $T$ is called $\mathcal I$-separable if
$T/\operatorname{Tr}T\in\operatorname{Sep}(\mathcal I)$. We denote
$T\in\operatorname{Sep}_+(\mathcal I)$. According to the definition of
$\mathfrak{I}$-separability, an $n$-mode state $\rho$ is
$\mathfrak{I}$-separable iff it can be written in the form
\begin{equation}\label{eq:rho-bisep}
      \rho = \sum_{\mathcal{I}\in\mathfrak{I}} p_\mathcal{I} \rho_\mathcal{I}, \quad 
      p_\mathcal{I} \geqslant 0, \quad
      \sum_{\mathcal{I}\in\mathfrak{I}} p_\mathcal{I} = 1, \quad 
      \rho_\mathcal{I} \in \Sep(\mathcal{I}).
\end{equation}
In contrast with this decomposition, one cannot remove the closure in
Eq.~\eqref{eq:SepI} since there are separable states that do not admit a
countable decomposition. 

Our quadratures are $r=(x_1,\ldots, x_n, p_1,\ldots, p_n)^T$. Their commutation relations read as 
\begin{equation}
 [r_j, r_k]=i\Omega_{jk},\quad
 \Omega=\begin{pmatrix}0&E\\-E&0\end{pmatrix},\label{eq:ccr}
\end{equation}
and covariance is given by
\begin{equation}
    \gamma_{jk}=\tfrac12\langle\{\Delta\hat r_j,\Delta\hat r_k\}\rangle,
 \quad \gamma_{\rm vac}=\tfrac{1}{2}E.
\end{equation}
We use the characteristic function
\begin{equation}
 \chi_{\varrho}(\xi)=\Tr(\varrho e^{i\xi^Tr}).
\end{equation}
A state $\varrho_G$ is called Gaussian of mean $\mu$ and covariance
$\gamma$ if its characteristic function reads as
\begin{equation}
 \chi_{\varrho_G}(\xi)=\exp(-\tfrac12\xi^T\gamma\xi+i\xi^T\mu).
 \label{eq:characteristic}
\end{equation}
The physical covariance condition is $\gamma+i\Omega/2\succcurlyeq 0$.
We use the standard finite-mode Gaussian existence and purity
characterization, with its exact reference and normalization specified
in Lemma~\ref{lem:cayley}. Gaussian states here include pure states and
states with pure normal modes.

The main result of this work is to prove the following theorem.
\begin{theorem}\label{thm:component} Let $\rho_G$ be a Gaussian state with
covariance $\gamma$. Suppose that, for a fixed partition $\mathcal I$, there is
a nonzero positive operator $T\in\operatorname{Sep}_+(\mathcal I)$ such that $T
\preccurlyeq \rho_G$. Then $\rho_G$ is separable across $\mathcal I$.
\end{theorem}

\begin{corollary}\label{thm:Sep}
Let $\mathfrak{I}$ be an arbitrary collection of arbitrary partitions of $n$
modes. For any Gaussian state $\rho_G$ the following implication is valid:
\begin{equation}
    \rho_G \in \Sep(\mathfrak{I}) \Rightarrow
    \rho_G \in \Sep(\mathcal{I})\ \text{for some}\ \mathcal{I}\in \mathfrak{I}.
\end{equation}
In other words, mixing states separable across different partitions creates no
new Gaussian states beyond the union of the corresponding fixed-partition
separable sets.
\end{corollary}
\noindent Note that the opposite implication trivially works for all states, not
only Gaussian.
\begin{proof}
Assume that Theorem~\ref{thm:component} has already been established. Suppose
that $\rho_G$ is Gaussian and belongs to $\Sep(\mathfrak{I})$. It admits an
exact convex decomposition given by Eq.~\eqref{eq:rho-bisep}. Terms with zero
weight may be omitted. Since the weights sum to one, choose a partition
$\mathcal{I}_0$ with $p_{\mathcal{I}_0}>0$, and put $T = p_{\mathcal{I}_0}
\rho_{\mathcal{I}_0}$. This is a nonzero positive operator separable across
$\mathcal{I}_0$. Moreover
\begin{equation}
    \rho_G - T = \sum_{\mathcal{I}_0\ne\mathcal{I}\in\mathfrak{I}} 
    p_\mathcal{I} \rho_\mathcal{I} \succcurlyeq 0,
\end{equation}
hence
\begin{equation}
    0\ne T \preccurlyeq \rho_G.
\end{equation}
Theorem~\ref{thm:component} now applies: a Gaussian state that dominates a nonzero
separable operator is itself separable across that operator's partition.
Therefore $\rho_G\in\Sep(\mathcal{I}_0)$.
\end{proof}
\noindent The Gaussian conjecture immediately follows from this corollary by
setting $\mathfrak{I} = \Pi_n$.
\begin{corollary}\label{thm:multipartite} For any $n\geqslant 3$, an $n$-mode
Gaussian state belongs to $\Bisep_n$ if and only if it belongs to $\Sep(\pi)$
for at least one fixed cut $\pi\in\Pi_n$. An equivalent formulation: a Gaussian
state is GME if and only if it is fully inseparable. Thus, the notions of full
inseparability and GME coincide for Gaussian states.
\end{corollary}
\noindent Another corollary is the absence of identical-copy GME activation for Gaussian
states.
\begin{corollary}\label{cor:copies} If $\rho_G$ is a biseparable Gaussian
state, then for every integer $m\geqslant 1$, $\rho_G^{\otimes m}$ is
biseparable when each original party holds all of its own copies. Indeed it is
separable across a fixed induced cut.
\end{corollary}
\begin{proof}
For simplicity, consider the tripartite case, the extension to the general case
is straightforward. Label the three nodes by $A$, $B$ and $C$. If $\rho_G$ is
Gaussian and biseparable, it is separable across at least one fixed cut, for
example as $A|BC$. Then every tensor power remains separable across
\begin{equation}
    A_1 \ldots A_m | B_1 C_1 \ldots B_m C_m.
\end{equation}
This is the induced cut: the same original division is applied to all copies. To
see the mechanism explicitly, first consider a finite separable mixture across a
general bipartition $X|Y$:
\begin{equation}
    \rho = \sum_i p_i \sigma^X_i \otimes \tau^Y_i, \quad
    p_i \geqslant 0, \quad
    \sum_i p_i = 1.
\end{equation}
Its two-copy state, after regrouping the subsystems by party, is
\begin{equation}
    \rho^{\otimes 2} = \sum_{ij} p_i p_j (\sigma^{X_1}_i \otimes \sigma^{X_2}_j) \otimes
    (\tau^{Y_1}_i \otimes \tau^{Y_2}_j).
\end{equation}
Every term is a product across $X_1 X_2 | Y_1 Y_2$. Therefore
$\rho^{\otimes2}$ is separable across that induced cut. The same expansion
works for any finite number of copies. The crucial point is that every copy has
the same separable bipartition.

As we discussed earlier, an infinite-dimensional separable state need not admit
an exact finite or countable product decomposition. Choose finite product
mixtures $\rho_j$, separable across the fixed bipartition $X|Y$, such that
$\|\rho_j - \rho_G\|_1 \to 0$. For each $j$, the preceding calculation shows
that $\rho_j^{\otimes m}$ is separable across the induced cut. Moreover
\begin{equation}
    \|\rho^{\otimes m}_j - \rho^{\otimes m}_G\|_1 \leqslant m\|\rho_j - \rho_G\|_1 \to 0.
\end{equation}
Thus $\rho^{\otimes m}_G$ is a trace-norm limit of states separable across
that same bipartition. Since the separable set is closed, the limit is also
separable across it and hence biseparable.
\end{proof}

The absence of identical-copy GME activation does not prevent two
\emph{different} biseparable Gaussian states from having a GME tensor product.
The distinction is that the two states need not share a separable bipartition.
This complements the identical-copy activation question studied in
Ref.~\cite{Baksova2025}. Consider three parties $A,B,C$. For $j=1,2$, the $j$th
input contains one mode $A_j$, one mode $B_j$, and one mode $C_j$. Let $\tau$ be
a one-mode Gaussian state, and let $\omega$ be an entangled two-mode Gaussian
state. Define, with subsystem labels specifying the tensor factors,
\begin{equation}
 \rho^{(1)}=\tau_{A_1}\otimes\omega_{B_1C_1},
 \qquad
 \rho^{(2)}=\tau_{B_2}\otimes\omega_{A_2C_2}.
 \label{app:eq:different-inputs}
\end{equation}
Thus $\rho^{(1)}$ is separable across $A_1|B_1C_1$, while $\rho^{(2)}$ is
separable across $B_2|A_2C_2$. Neither is GME. They are different states
on the labelled tripartite system: the entangled pair occupies different
parties.

Now form their tensor product and group the modes party-wise:
\begin{equation}
 \rho=\rho^{(1)}\otimes\rho^{(2)},
 \qquad A=A_1A_2,\quad B=B_1B_2,\quad C=C_1C_2.
 \label{app:eq:grouped-product}
\end{equation}
The resulting state has six modes but three parties. We claim that the Gaussian
state in Eq.~\eqref{app:eq:grouped-product} is genuinely tripartite entangled
across the parties $A,B,C$. In fact, suppose that $\rho$ were separable across
$A|BC$. Discarding the entire first input is a local operation for this cut, so
the remaining state $\rho^{(2)}$ would be separable across $A_2|B_2C_2$.
Discarding $B_2$ would then make $\omega_{A_2C_2}$ separable across $A_2|C_2$, a
contradiction. Similarly, separability of $\rho$ across $B|AC$ would imply
separability of $\omega_{B_1C_1}$ after discarding the second input and mode
$A_1$. Separability across $C|AB$ would also imply separability of
$\omega_{B_1C_1}$ under the same reduction. Hence $\rho$ is entangled across all
three bipartitions. A tensor product of Gaussian states is Gaussian. The
Gaussian full-inseparability/GME equivalence following from Theorem~1 therefore
shows that $\rho$ is GME. The last step is essential for mixed states:
entanglement across every cut does not, for an arbitrary non-Gaussian mixed
state, by itself imply GME. Also, the conclusion here is genuine
\emph{tripartite} entanglement, not genuine six-mode entanglement when all six
modes are treated as separate parties.

We stress that Theorem 1 is a one-way implication. It says that a Gaussian state
dominating a nonzero separable positive operator is itself separable across that
operator's partition. It does not assert that every positive operator dominated
by a separable Gaussian state must itself be separable. In fact, we can
construct an example of a tripartite fully separable Gaussian state $\tau_3$
which is an equal-weight mixture of three entangled states. Define
\begin{equation}\label{eq:tau1}
    \tau = \frac{1}{2}\sum^{+\infty}_{k=0}2^{-k}\ket k \bra k,
\end{equation}
and set $\tau_3 = \tau^{\otimes 3}$. Throughout this work, $|n\rangle$ denotes
the standard $n$-photon Fock state, with $|0\rangle$ being the vacuum. Multimode
Fock states are denoted by $|n_1\ldots n_m\rangle$. This is a mixed Gaussian
state, fully separable and even factorizable by construction. There is a
decomposition
\begin{equation}\label{eq:tau}
    \tau_3 = \frac{1}{3}(\sigma_1 + \sigma_2 + \sigma_3),
\end{equation}
where 
\begin{equation}
\begin{split}
    \sigma_1\in\Sep(1|23), &\quad \sigma_1\text{ is NPT across }2|13\text{ and }3|12,\\
    \sigma_2\in\Sep(2|13), &\quad \sigma_2\text{ is NPT across }1|23\text{ and }3|12,\\
    \sigma_3\in\Sep(3|12), &\quad \sigma_3\text{ is NPT across }1|23\text{ and }2|13.
\end{split}
\end{equation}
Moreover, the covariance matrices
$\gamma_{\sigma_1},\gamma_{\sigma_2},\gamma_{\sigma_3}$ of the decomposition
states are mutually different and none equals $\gamma_{\tau_3}$, while their
arithmetic mean equals $\gamma_{\tau_3}$. Note that Eq.~\eqref{eq:tau} expresses
the equality of the full density operators, not just the covariance matrices. 

There is an even more extreme decomposition of this Gaussian state:
\begin{equation}
    \tau_3 = \frac{1}{2}(\sigma_+ + \sigma_-),
\end{equation}
where $\sigma_\pm$ are both GME. Although separable finite-dimensional states
are well known to admit decompositions into entangled and even GME states,
Ref.~\cite{Carvacho2017GHZGeometry}, we are not aware of an explicit finite
decomposition of a fully separable finite-mode Gaussian state into GME states.

\section{Applications}

Recent independent work \cite{Yang2026GaussianRigidity} proves
Corollary~\ref{thm:Sep} by a different method. Another main result of that work
reads as follows: if $\Ran\rho_G^{1/2}$ of a Gaussian state $\rho_G$ contains a
nonzero vector factorizable across some partition $\mathcal{I}$, then
$\rho_G\in\Sep(\mathcal{I})$. This statement is also a simple consequence of
Theorem~\ref{thm:component}. In fact, let
\begin{equation}
    0\neq\ket\phi\in\Ran\rho_G^{1/2}
\end{equation}
be product across $\mathcal{I}$. By definition of the range, there is a
vector $\ket\psi$ such that
\begin{equation}
    \ket\phi=\rho_G^{1/2}\ket\psi.
\end{equation}
For every vector $\ket u$, Cauchy--Schwarz gives
\begin{equation}
    |\langle u\ket\phi|^2 = 
    |\langle \rho_G^{1/2}u\ket\psi|^2 \leqslant
    \|\psi\|^2\,\bra u\rho_G\ket u.
\end{equation}
Hence, as an operator inequality, we have
\begin{equation}
    \ket\phi\bra\phi \preccurlyeq \|\psi\|^2\rho_G.
\end{equation}
Since $\ket\phi\neq0$, also $\ket\psi\neq0$.  Define
\begin{equation}
    T=\frac{\ket\phi\bra\phi}{\|\psi\|^2}.
\end{equation}
We again obtain a hypothesis for Theorem~\ref{thm:component}
\begin{equation}
    0\neq T\preccurlyeq\rho_G.
\end{equation}
Because $\ket\phi$ is factorizable across $\mathcal{I}$, the rank-one operator
$T$ is separable across $\mathcal{I}$. The domination theorem now gives
\begin{equation}
    \rho_G\in\Sep(\mathcal{I}).
\end{equation}
Next we present several more substantial applications of the main theorem.

\subsection{Entanglement-breaking domination rigidity for Gaussian channels}
\label{app:subsec:channels}

We consider channels in the Schr\"odinger picture, acting on trace-class
operators. All maps below are bounded and completely positive (CP), with normal
adjoints. A channel is trace preserving. A \emph{subchannel} $\Psi$ of a channel
$\Phi$ satisfies
\begin{equation}
 0\preccurlyeq_{\mathrm{CP}}\Psi\preccurlyeq_{\mathrm{CP}}\Phi,
 \label{app:eq:cp-order}
\end{equation}
meaning that both $\Psi$ and $\Phi-\Psi$ are CP. Such a $\Psi$ is automatically
trace nonincreasing, and may represent one outcome of an instrument whose
unconditioned channel is $\Phi$. This order is stronger than positivity on
uncorrelated input states.

A channel is \emph{entanglement breaking} if, for every reference system
$R$ and every input state on $R\otimes A$, its action on $A$ produces a
state separable across $R|B$. For a trace-decreasing CP map, we use the
same definition with an unnormalized separable output, allowing zero.
The reference may be arbitrary. We use the infinite-dimensional
framework of Refs.~\cite{Holevo2005,Holevo2008EB}.

\begin{proposition}
\label{app:prop:eb-rigidity}
Let $\Phi:A\to B$ be a Gaussian channel. If $\Phi$ has a nonzero
entanglement-breaking subchannel $\Psi$, then $\Phi$ is entanglement breaking.
The subchannel $\Psi$ need not be Gaussian.
\end{proposition}

\begin{proof}
Let the input have $m$ modes and introduce a reference $R$ with $m$
modes. Fix $0<q<1$ and prepare
\begin{equation}
 |\Omega_q\rangle
 =(1-q)^{m/2}\sum_{\boldsymbol n\in\mathbb N_0^m}
 q^{|\boldsymbol n|/2}
 |\boldsymbol n\rangle_R|\boldsymbol n\rangle_A,
 \qquad |\boldsymbol n|=n_1+\cdots+n_m.
 \label{app:eq:probe}
\end{equation}
This is a product of $m$ two-mode squeezed vacua. It is a normalized
finite-energy Gaussian state, with faithful thermal marginal
\begin{equation}
 \theta_q=(1-q)^m\sum_{\boldsymbol n\in\mathbb N_0^m}
 q^{|\boldsymbol n|}|\boldsymbol n\rangle\langle\boldsymbol n|.
 \label{app:eq:probe-marginal}
\end{equation}
Here \emph{faithful} means that the kernel is zero. No infinite-squeezing
limit or unnormalizable maximally entangled vector is needed.

Apply the two maps to the same probe:
\begin{equation}
 J_q(\Phi)=(\operatorname{id}_R\otimes\Phi)
                (|\Omega_q\rangle\langle\Omega_q|), \quad
 T_q=(\operatorname{id}_R\otimes\Psi)
                (|\Omega_q\rangle\langle\Omega_q|).
 \label{app:eq:channel-probe-outputs}
\end{equation}
The state $J_q(\Phi)$ is Gaussian. Complete positivity of both maps in
Eq.~\eqref{app:eq:cp-order} implies $0\preccurlyeq T_q\preccurlyeq J_q(\Phi)$.
The entanglement-breaking property of $\Psi$ makes $T_q$ separable across $R|B$.
We must also check that $T_q\ne0$. Let $E_\Psi=\Psi^*(E_B)$, where $E_B$ is the
identity on the output space. This is the effect associated with the branch.
Since $\Psi$ is nonzero and positive, $E_\Psi$ is a nonzero positive operator.
All diagonal weights of $\theta_q$ are strictly positive, and therefore
\begin{equation}
 \operatorname{Tr}T_q
 =\operatorname{Tr}(E_\Psi\theta_q)>0.
 \label{app:eq:nonzero-branch}
\end{equation}
Consequently Theorem~1 gives
$J_q(\Phi)\in\operatorname{Sep}(R|B)$.

Finally, separability of this one faithful-probe output implies that $\Phi$ is
entanglement breaking. This is the faithful-state version of the Choi
characterization underlying the infinite-dimensional channel construction in
Refs.~\cite{Holevo2005,Holevo2008EB}. For completeness, it follows directly as
follows. Enumerate the Fock basis and write
$|\Omega_q\rangle=\sum_j\sqrt{p_j}|j\rangle_R|j\rangle_A$ with $p_j>0$. For an
arbitrary reference $R'$ and a vector with finite input support
\begin{equation}
 |\chi\rangle_{R'A}=\sum_{j=1}^{L}|u_j\rangle_{R'}|j\rangle_A,
 \qquad
 F=\sum_{j=1}^{L}\frac{|u_j\rangle_{R'}\langle j|_R}{\sqrt{p_j}},
 \label{app:eq:finite-filter}
\end{equation}
the finite-rank operator $F$ is bounded and
$|\chi\rangle=(F\otimes E_A)|\Omega_q\rangle$. Hence
\begin{equation}
 (\operatorname{id}_{R'}\otimes\Phi)(|\chi\rangle\langle\chi|)
  =(F\otimes E_B)J_q(\Phi)(F^\dagger\otimes E_B).
 \label{app:eq:filtered-choi}
\end{equation}
The right-hand side is separable, since the filter acts only on the
reference. Every pure input is a norm limit of vectors with finite input
support. Trace-norm continuity of the channel and closedness of the
separable set extend the conclusion to every pure input, and then to
mixed inputs by their spectral decompositions. Thus $\Phi$ is
entanglement breaking.
\end{proof}

A consequence of this proposition is that a non-entanglement-breaking Gaussian
channel cannot contain even a small nonzero entanglement-breaking branch. In
particular, it has no convex decomposition
\begin{equation}
 \Phi=p\Phi_{\mathrm{EB}}+(1-p)\Xi,
 \qquad 0<p\leqslant1,
 \label{app:eq:eb-mixture-exclusion}
\end{equation}
where $\Phi_{\mathrm{EB}}$ is an entanglement-breaking channel and $\Xi$
is any channel. The assumption that the \emph{total} channel is Gaussian
is essential; its proposed components need not be Gaussian.

In fact the same probe $|\Omega_q\rangle$ detects entanglement in every
nonzero subchannel of a non-entanglement-breaking Gaussian channel.
Otherwise its separable branch output would be a nonzero separable
operator dominated by $J_q(\Phi)$, and Theorem~1 would imply that $\Phi$
is entanglement breaking. Thus recording an instrument outcome cannot
produce a separable conditional probe output with nonzero probability.
This statement concerns CP branches or coarse-grained outcomes, not
formal conditional densities at zero-probability points of a continuous
measurement.

A further immediate consequence is that no Kraus representation of such
a channel can contain a nonzero bounded rank-one operator. Indeed,
$L=|a\rangle\langle b|$ would define the entanglement-breaking subchannel
$\Psi(X)=LXL^\dagger$. This observation concerns rank one; it does not
exclude every finite rank by the same argument.

A simple example is a thermal attenuator. For a one-mode attenuator with
$0<\eta<1$ and thermal environment photon number $n_{\mathrm{th}}\geqslant0$, the
covariance transformation is
\begin{equation}
 \gamma\longmapsto\eta\gamma
 +(1-\eta)\bigl(n_{\mathrm{th}}+\tfrac12\bigr)E_2.
 \label{app:eq:attenuator}
\end{equation}
The established Gaussian-channel criterion gives the entanglement-breaking
threshold~\cite{Holevo2008EB}
\begin{equation}
 n_{\mathrm{th}}\geqslant\frac{\eta}{1-\eta}.
 \label{app:eq:attenuator-threshold}
\end{equation}
Below this threshold, Proposition~\ref{app:prop:eb-rigidity} excludes
\emph{all} nonzero entanglement-breaking subchannels, including
non-Gaussian ones. For example, an attenuator with $\eta=1/2$ and
$n_{\mathrm{th}}<1$ has no such branch.

\subsection{No exact preparation of entangled Gaussian states from faithful separable noise}
\label{app:subsec:preparation}

Theorem~1 also obstructs exact state preparation. A positive separable
contribution to the input remains separable after a separability-preserving
operation. If that contribution survives and the output is Gaussian,
Theorem~1 forces the output to be separable. For deterministic operations
survival is automatic. For postselected operations it is guaranteed by
faithfulness of the separable contribution.

An operation $\Lambda$ is a CP, trace-nonincreasing map. We call it
\emph{separability preserving} for specified input and output partitions
if it maps the input separable cone into the output separable cone. In
particular, every branch of a protocol using local operations and
classical communication (LOCC) has this property, when the partitions
refer to the same parties. The assumption is on the selected branch
itself, not merely on the channel obtained by ignoring its outcome.
The operation need not be Gaussian.

\begin{proposition}
\label{app:prop:no-go}
Let $\rho$ be a state and suppose that, for some $\varepsilon>0$,
\begin{equation}
 \rho\succcurlyeq\varepsilon\sigma_{\mathrm{sep}},
 \qquad
 \sigma_{\mathrm{sep}}\in\operatorname{Sep}(\mathcal I),
 \qquad
 \ker\sigma_{\mathrm{sep}}=\{0\}.
 \label{app:eq:noise-bound}
\end{equation}
For any finite $N\geqslant1$, group the $N$ input copies party-wise, giving the
induced input partition $\mathcal I^{(N)}$. Let $\Lambda$ preserve separability
from this partition to a specified output partition $\mathcal I_{\mathrm{out}}$
of output modes. If the success probability $p_{\mathrm
s}=\operatorname{Tr}\Lambda(\rho^{\otimes N})$ is positive and the normalized
output
\begin{equation}
 \rho_{\mathrm{out}}=
 \frac{\Lambda(\rho^{\otimes N})}{p_{\mathrm s}}
 \label{app:eq:normalized-output}
\end{equation}
is Gaussian, then it is separable across the output partition.
\end{proposition}

\begin{proof}
Set $\Delta=\rho-\varepsilon\sigma_{\mathrm{sep}}\succcurlyeq0$. Expanding
$(\varepsilon\sigma_{\mathrm{sep}}+\Delta)^{\otimes N}$ as a tensor
product shows that every term is positive. In particular,
\begin{equation}
 \rho^{\otimes N}\succcurlyeq
 \varepsilon^N\sigma_{\mathrm{sep}}^{\otimes N}.
 \label{app:eq:noise-tensor-bound}
\end{equation}
Applying the positive map $\Lambda$ and dividing by $p_{\mathrm s}$
gives
\begin{equation}
 \rho_{\mathrm{out}}\succcurlyeq T,
 \qquad
 T=\frac{\varepsilon^N}{p_{\mathrm s}}
       \Lambda(\sigma_{\mathrm{sep}}^{\otimes N})
       \in\operatorname{Sep}_+(\mathcal I_{\mathrm{out}}).
 \label{app:eq:surviving-piece}
\end{equation}
The tensor power of a separable state is separable across the induced
partition, and $\Lambda$ preserves this cone.

It remains to show that $T$ is nonzero. Let
$E_\Lambda=\Lambda^*(E_{\mathrm{out}})$. Then
$0\preccurlyeq E_\Lambda\preccurlyeq E_{\mathrm{in}}$. Since
$p_{\mathrm s}>0$, the effect $E_\Lambda$ is nonzero.
The state $\sigma_{\mathrm{sep}}^{\otimes N}$ is faithful, so
\begin{equation}
 \operatorname{Tr}\Lambda(\sigma_{\mathrm{sep}}^{\otimes N})
 =\operatorname{Tr}(E_\Lambda\sigma_{\mathrm{sep}}^{\otimes N})>0.
 \label{app:eq:noise-survival}
\end{equation}
For clarity, if $\sigma_{\mathrm{sep}}^{\otimes N} =\sum_j s_j|e_j\rangle\langle
e_j|$ with every $s_j>0$, a zero trace would imply $\langle
e_j|E_\Lambda|e_j\rangle=0$ for all $j$. Positivity would then give
$E_\Lambda^{1/2}|e_j\rangle=0$ for a complete basis and hence $E_\Lambda=0$, a
contradiction.

We have obtained $0\ne T\preccurlyeq\rho_{\mathrm{out}}$ with $T$ separable.
Since $\rho_{\mathrm{out}}$ is Gaussian, Theorem~1 proves the claim.
\end{proof}

The physical reason is that a faithful noise state assigns positive probability
to every nonzero success effect. No successful filtering outcome can reject this
noise perfectly. A nonzero separable contribution therefore survives, and an
entangled Gaussian state cannot contain such a contribution. The hypothesis is
an \emph{order bound}, not just the statement that $\rho$ is faithful. In
infinite dimensions, faithfulness means a trivial kernel; it does not mean a
strictly positive lower bound proportional to the identity, and it does not by
itself imply Eq.~\eqref{app:eq:noise-bound}. A sufficient and experimentally
natural form of the hypothesis is an unflagged mixture
\begin{equation}
 \rho=(1-\varepsilon)\omega
       +\varepsilon\sigma_{\mathrm{sep}},
 \qquad 0<\varepsilon<1,
 \label{app:eq:contaminated-resource}
\end{equation}
with faithful separable $\sigma_{\mathrm{sep}}$ and arbitrary resource
state $\omega$. Neither state needs to be Gaussian. An accessible
classical label identifying the two terms is a different input resource;
then the hypothesis must be checked on the enlarged state including
that label.

Consider an explicitly entangled noisy input. Take the thermal state $\tau$ from
Eq.~\eqref{eq:tau1} and let
\begin{equation}
 |\psi\rangle=\frac{|00\rangle+|11\rangle}{\sqrt2},
 \qquad
 \rho_\varepsilon=(1-\varepsilon)|\psi\rangle\langle\psi|
                    +\varepsilon\tau\otimes\tau.
 \label{app:eq:noisy-entangled-example}
\end{equation}
The state $\tau\otimes\tau$ is faithful and separable. The partial
transpose of $\rho_\varepsilon$ has the principal block
\begin{equation}
 \begin{pmatrix}
 \varepsilon/8&(1-\varepsilon)/2\\
 (1-\varepsilon)/2&\varepsilon/8
 \end{pmatrix}
 \label{app:eq:noisy-input-pt}
\end{equation}
on $\operatorname{span}\{|01\rangle,|10\rangle\}$, with smaller eigenvalue
$(5\varepsilon-4)/8$. Thus $\rho_\varepsilon$ is NPT entangled whenever
$0<\varepsilon<4/5$. Nevertheless, no finite number of copies can be converted,
by an LOCC branch of nonzero probability, into an \emph{exactly} Gaussian
entangled state, even a mixed one. For a concrete value, $\varepsilon=1/2$ gives
the negative eigenvalue $-3/16$ in Eq.~\eqref{app:eq:noisy-input-pt}. For
several parties the same proposition applies to any fixed partition. If the
noise is faithful and fully separable, every exactly Gaussian output of a
successful LOCC branch must be fully separable.

This is an exact, finite-copy obstruction. It does not prohibit approximate
preparation, a limit of protocols with increasing copy number, or distillation
into non-Gaussian entangled states. In particular it does not conflict with
iterative Gaussification, which approaches Gaussian
outputs~\cite{PhysRevA.67.062320}. The surviving separable weight
\begin{equation}
 w_{\mathrm{surv}}
 =\frac{\varepsilon^N
     \operatorname{Tr}\Lambda(\sigma_{\mathrm{sep}}^{\otimes N})}
        {p_{\mathrm s}}>0
 \label{app:eq:surviving-weight}
\end{equation}
need not have a positive lower bound uniform over protocols or copy
numbers. The proof therefore gives no universal approximation-error
bound.

The use of a surviving free component to obstruct exact purification is related
to the general resource-purification approach of Fang and
Liu~\cite{PRXQuantum.3.010337}. The additional input here is
Theorem~\ref{thm:component}: it rules out a separable component not only for
pure entangled targets but for \emph{every entangled finite-mode Gaussian
target}, including mixed states. Gaussianity is required of the target, not of
the local operations.

\subsection{A connected set of Gaussian marginals}

Let $\rho$ be an arbitrary $n$-mode state, and write
\begin{equation}
 \rho_{ij}=\operatorname{Tr}_{\overline{ij}}\rho
 \label{gmar:eq:reduction}
\end{equation}
for its bipartite reductions, where we trace out all modes except the modes $i$
and $j$. Construct a graph $G=(V,\mathcal E)$, with one vertex for each mode,
such that there is an edge $\{i,j\}\in\mathcal E$ if the state $\rho_{ij}$ is
Gaussian and entangled across $i|j$. It is not necessary to include every such
pair in the graph. In particular, a connected subset of the available pairs is
enough.

\begin{proposition}
\label{prop:connected}
If $G$ is connected, then $\rho$ is genuinely multipartite entangled. The
global state $\rho$ need not be Gaussian.
\end{proposition}

\begin{proof}
Suppose that $\rho$ were biseparable. Applying Eq.~\eqref{eq:rho-bisep} to the
finite set of bipartitions gives
\begin{equation}
 \rho=\sum_\pi p_\pi\rho_\pi,
 \qquad p_\pi\geqslant 0,\quad \sum_\pi p_\pi=1,
 \qquad \rho_\pi\in\operatorname{Sep}(\pi).
 \label{eq:bisep-decomp}
\end{equation}
Choose a term with $p_{\pi_0}>0$, write $\pi_0=X|\overline X$, and set
$T=p_{\pi_0}\rho_{\pi_0}$. Then $0\ne T\preccurlyeq\rho$. Since $G$ is
connected, some edge $\{i,j\}$ crosses this cut, which means that $i\in X$ and
$j\in\overline X$. Taking its two-party reduction gives
\begin{equation}
 T_{ij}:=\operatorname{Tr}_{\overline{ij}}T,
 \qquad 0\preccurlyeq T_{ij}\preccurlyeq\rho_{ij},
 \qquad \operatorname{Tr}T_{ij}=p_{\pi_0}>0.
 \label{gmar:eq:reduced-domination}
\end{equation}
Partial tracing preserves separability across the surviving sides of a cut, so
$T_{ij}$ is separable across $i|j$. The Gaussian state $\rho_{ij}$ therefore
dominates a nonzero separable positive operator. Theorem~1 makes $\rho_{ij}$
separable, contradicting the choice of the edge. Hence $\rho$ cannot be
biseparable.
\end{proof}

The proof uses more than the fact that an entangled marginal obstructs
separability of the whole state across the corresponding cut. It excludes
\emph{each possible biseparable contribution}, even when the cut varies from
term to term. This is why the conclusion is GME, rather than merely entanglement
across every fixed cut.

Connected pairwise entanglement is not sufficient for GME for general
mixed states. For comparison, consider the three-qubit state
\begin{equation}
 \eta=
 \frac12|\Phi^+\rangle\langle\Phi^+|_{12}\otimes|0\rangle\langle0|_3
 +\frac12|0\rangle\langle0|_1\otimes|\Phi^+\rangle\langle\Phi^+|_{23},
 \quad
 |\Phi^+\rangle=\frac{|00\rangle+|11\rangle}{\sqrt2}.
 \label{gmar:eq:qubit-counterexample}
\end{equation}
This state is biseparable by construction. Nevertheless, the partial
transpose of $\eta_{12}$ has the principal block
\begin{equation}
 \begin{pmatrix}1/4&1/4\\[1mm]1/4&0\end{pmatrix}
 \quad\text{on }\operatorname{span}\{|01\rangle,|10\rangle\},
 \label{gmar:eq:qubit-negative-block}
\end{equation}
whose determinant is $-1/16$. Hence $\eta_{12}$ is entangled, and symmetry under
exchange of modes $1$ and $3$ gives the same conclusion for $\eta_{23}$.
Embedding these qubits into Fock levels $|0\rangle,|1\rangle$ also makes a
continuous-variable counterexample without Gaussian marginals. The extra force
of Proposition~\ref{prop:connected} comes from the inability of an entangled
Gaussian marginal to contain any positive separable contribution.

In addition, one can construct a three-mode state that is non-Gaussian globally,
but has three entangled Gaussian two-mode marginals. This shows that
Theorem~\ref{thm:component} can certify GME not for Gaussian states only. The
construction of such state is given in Appendix~\ref{app:C3}. 

As an example, a spanning tree of a graph with $n$ vertices contains $n-1$ edges
and is connected. Thus $n-1$ appropriately chosen entangled Gaussian two-party
marginals suffice. For three parties, the particularly simple implication is
\begin{equation}
 \rho_{12}\text{ and }\rho_{23}\text{ Gaussian and entangled}
 \quad\Longrightarrow\quad \rho_{123}\text{ is GME}.
\end{equation}
No assumption is needed about $\rho_{13}$.

The graph also gives information when it is disconnected. The same argument
proves that any nonzero $T\in\operatorname{Sep}_+(\pi)$ with $T\preccurlyeq\rho$
must keep the endpoints of every edge in the same block of $\pi$. Otherwise its
reduction to a crossing edge contradicts Theorem~\ref{thm:component}. Following
paths, each connected component of $G$ must lie entirely inside one block.

Let $d(\rho)$ be the entanglement depth: the least integer $k$ such that $\rho$
belongs to the closed convex hull of states separable across partitions with at
most $k$ parties in each block. If the largest connected component of $G$ has
$s$ vertices, then
\begin{equation}
 d(\rho)\geqslant s.
\end{equation}
Indeed, Eq.~\eqref{eq:rho-bisep} would turn a $k$-producible decomposition with
$k<s$ into nonzero positive components, none of which could accommodate that
connected component inside one block. Here depth is counted in the specified
parties, not in the number of optical modes.

\subsection{Qualitative robustness under errors in the marginals}
\label{gmar:subsec:robustness}

The exact Gaussian-marginal criterion also has a qualitative stability
property. For fixed target marginals, sufficiently small trace-norm errors
do not destroy the conclusion, even when the actual marginals and global
state are non-Gaussian.

\begin{proposition}
\label{prop:robustness}
Let $G=(V,\mathcal E)$ be connected, and let
$\{\eta_{ij}:\{i,j\}\in\mathcal E\}$ be a compatible collection of
entangled Gaussian two-party states. There exists $\varepsilon>0$ such
that every global state $\rho$ satisfying
\begin{equation}
 \max_{\{i,j\}\in\mathcal E}
       \|\rho_{ij}-\eta_{ij}\|_1<\varepsilon
 \label{gmar:eq:tolerance}
\end{equation}
is GME. The tolerance may depend on the target marginals and the graph.
\end{proposition}
 The proof is given in the Appendix~\ref{app:E}. This proposition establishes
the \emph{existence} of a positive tolerance, not an explicit numerical error
bar. No target-independent tolerance is asserted as the target marginals
approach the separable boundary. Also, closeness of covariance matrices alone
does not establish Eq.~\eqref{gmar:eq:tolerance} for arbitrary non-Gaussian
states. Turning this result into an experimental error budget requires
additional information controlling the distance between the marginal density
operators.

GME certification from incomplete marginal information has been developed in
other continuous-variable settings. Partially blind witnesses can certify
Gaussian GME from selected separable two-mode marginals, and minimal covariance
criteria can use tree-structured marginal
data~\cite{PhysRevA.106.062410,Leskovjanova2025minimalcriteria}. The present
criterion is different rather than uniformly stronger: it requires
\emph{entangled Gaussian} marginals, but makes no global Gaussianity assumption
and certifies every compatible global state. The argument here uses
positive-operator domination to exclude every biseparable contribution to a
compatible global state, rather than optimizing a witness over covariance
matrices.

The rest of the work is devoted to the proof of Theorem~\ref{thm:component}.

\section{Husimi curvature of separable operators}\label{sec:curvature}

For a complex symmetric matrix $K$ define
\begin{equation}
 \Realif(K)=
 \begin{pmatrix}\operatorname{Re}K&-\operatorname{Im}K\\
                -\operatorname{Im}K&-\operatorname{Re}K\end{pmatrix}.
 \label{eq:realification}
\end{equation}
Simple properties of matrices of this form are expressed by the following lemma.
\begin{lemma}\label{lem:realification}
For complex symmetric $K$, $A=\Realif(K)$ is real symmetric and
\begin{equation}
 A\Omega=-\Omega A, \quad \Omega^TA\Omega=-A,\quad \|A\|=\|K\|.
 \label{eq:sign-reversal}
\end{equation}
In particular, if $A\preccurlyeq cE$, $c\geqslant 0$, then $-cE\preccurlyeq
A\preccurlyeq cE$ and $\|K\|\leqslant c$.
\end{lemma}
\noindent The standard norm of an operator $A$ is defined as 
\begin{equation}\label{eq:st-norm}
    \|A\| = \sup_{\|v\|=1} \|Av\|,
\end{equation}
where $\|v\|$ denotes the standard Euclidean norm on $\mathbb R^{2n}$ or
$\mathbb C^n$, as appropriate.
\begin{proof}
Block multiplication proves the first two identities. Under the real isometry
$(x,p)^{\mathsf T}\mapsto x+ip$, the action of $A$ is $u\mapsto\overline{Ku}$.
Complex conjugation preserves norm, so we have
\begin{equation}
    \left\|\Realif(K)
    \begin{pmatrix}
        x \\
        p
    \end{pmatrix}
    \right\|_{\R^{2n}} = 
    \|K(x+ip)\|_{\C^{n}},
\end{equation}
proving the operator-norm identity. Evaluating $A\preccurlyeq cE$ on $\Omega v$,
and using $\Omega^TA\Omega=-A$ together with $\|\Omega v\|=\|v\|$, gives
$-A\preccurlyeq cE$. The real symmetric spectral theorem gives the norm bound.
\end{proof}
The next result is an important construction in the proof of
Theorem~\ref{thm:component}.
\begin{lemma}[Cayley construction]\label{lem:cayley} Let $K^{\mathsf T}=K$, with
$\|K\|<1$, and put $A=\Realif(K)$. Then the matrix $\gamma$ defined via
\begin{equation}
 \gamma=(E-A)^{-1}-\tfrac{1}{2}E
       =\tfrac12(E+A)(E-A)^{-1}
 \label{eq:cayley}
\end{equation}
is the covariance matrix of a (zero-mean) pure Gaussian state. If $K$ is block
diagonal across $\mathcal{I}$, that state is a $\mathcal{I}$-product state.
\end{lemma}

\begin{proof}
The real symmetric matrix $A$ has spectrum in $(-1,1)$. Thus $L=(E+A)(E-A)^{-1}$
is symmetric positive definite. Using the equality $A\Omega=-\Omega A$ and
moving $\Omega$ across the factors gives
\begin{equation}
  L\Omega=\Omega L^{-1},
 \qquad L\Omega L=\Omega,\qquad \gamma\Omega\gamma=\tfrac{1}{4}\Omega.
\end{equation}
This also implies physicality, not only positivity of $\gamma$. If
$Lu=\lambda u$, then $L\Omega u=\lambda^{-1}\Omega u$. The spectral theorem
consequently gives $L^{1/2}\Omega=\Omega L^{-1/2}$, and hence
\begin{equation}
 \gamma+i\Omega/2
   =\tfrac12 L^{1/2}(E+i\Omega)L^{1/2}\succcurlyeq 0.
\end{equation}
Here $i\Omega$ is Hermitian with eigenvalues $1$ and $-1$.

We now apply the standard result: every real covariance $\gamma$ with
$\gamma+i\Omega/2\succcurlyeq 0$ is realized by a normalized Gaussian state,
and this state is pure precisely when $\gamma\Omega\gamma=\Omega/4$,
Ref.~\cite{WernerWolf2001}. If $A$ is partition-block-diagonal, so are $L$
and $\gamma$. Each local block satisfies the same physicality and purity
conditions. Construct its corresponding (zero-mean) pure Gaussian state. Their
tensor product has the characteristic function with covariance $\gamma$, so by
uniqueness it is the global Gaussian just constructed.
\end{proof}

The main object that helps us to establish Theorem~\ref{thm:component} is the
Husimi function. For $r=(x,p)\in\R^{2n}$ put $\alpha=(x+ip)/\sqrt2$ and use
normalized coherent states
\begin{equation}
 \ket{\alpha}=e^{-\|\alpha\|^2/2}\sum_{m\in\N^n}
                   \frac{\alpha^{m}}{\sqrt{m!}}\ket{m}.
\end{equation}
Husimi function of $A$ is defined as
\begin{equation}
    q_{A}(r)=\bra{\alpha} A\ket{\alpha} = 
    \left\langle\frac{x+ip}{\sqrt{2}}\middle|A\middle|\frac{x+ip}{\sqrt{2}}\right\rangle. 
    \label{eq:husimi}
\end{equation}
In fact, we omit the prefactor of the standard definition; our function $q_\rho(r)$
has the property
\begin{equation}
    \int_{\R^{2n}} q_\rho(r) \, dr = (2\pi)^n.
\end{equation}
This prefactor plays no role in the proofs below, so there is no reason to
introduce it and carry throughout all derivations.

The central result in the proof of Theorem~\ref{thm:component} is the following
proposition.
\begin{proposition}\label{prop:separable-curvature} If $T$ is a
$\mathcal{I}$-separable (not necessarily normalized) operator and $q_T(r)>0$,
there exists a complex symmetric $\mathcal{I}$-block-diagonal $K_T(r)$
satisfying
\begin{equation}
    \Hess\log q_T(r)\succcurlyeq-E+\Realif(K_T(r)).
    \label{eq:separable-curvature}
\end{equation}
\end{proposition}

\begin{proof}[Proof of Theorem~\ref{thm:component}]
Assume that Proposition~\ref{prop:separable-curvature} has been established. We
show how to prove Theorem~\ref{thm:component}. The operator inequality
$0\preccurlyeq T\preccurlyeq\rho_G$ gives
\begin{equation}
  0 \leqslant q_T(r) \leqslant q_G(r)
\end{equation}
for every $r$. Moreover, $T\ne0$ implies that $q_T$ is positive somewhere.

Because $\rho_G$ is Gaussian, its Husimi function has the form
\begin{equation}
  q_G(r)=\sqrt{\det(B)}e^{-\tfrac12(r-\mu)^TB(r-\mu)},
 \quad B=(\gamma+E/2)^{-1}\succ 0.
\end{equation}
Consequently
\begin{equation}
  \nabla^2 \log q_G(r) = -B.
\end{equation}
The important feature is that this logarithmic Hessian is constant throughout
phase space.

Separability of $T$ supplies the other ingredient.
Proposition~\ref{prop:separable-curvature} says that, at every point where
$q_T(r)>0$, there is a complex symmetric matrix $K(r)$, block diagonal
across $\mathcal I$, such that
\begin{equation}
  \nabla^2 \log q_T(r) \succcurlyeq -E + \Realif(K(r)).
\end{equation}
We now want to combine the upper bound $q_T\leqslant q_G$ with the curvature lower
bound. We cannot simply differentiate the inequality $q_T\leqslant q_G$ and conclude
an inequality between their Hessians. Instead, we compare them at a maximum of
their ratio. The ratio $q_T/q_G$ need not attain its supremum. To ensure that
a maximum exists, introduce
\begin{equation}
  F_\varepsilon(r) = e^{-\varepsilon\|r\|^2} \frac{q_T(r)}{q_G(r)}
\end{equation}
for $\varepsilon > 0$. Since $q_T/q_G\leqslant 1$
\begin{equation}
  0 \leqslant F_\varepsilon(r) \leqslant e^{-\varepsilon\|r\|^2}.
\end{equation}
Thus $F_\varepsilon$ tends to zero at infinity. It is continuous and positive
somewhere, so it attains a positive global maximum at some point
$r_\varepsilon$. In particular, $q_T(r_\varepsilon)>0$.

At this point, the Hessian of $\log F_\varepsilon$ is negative semidefinite:
\begin{equation}
  \nabla^2 \log F_\varepsilon(r_\varepsilon) \preccurlyeq 0.
\end{equation}
But
\begin{equation}
  \log F_\varepsilon(r) = -\varepsilon\|r\|^2 + \log q_T(r) - \log q_G(r).
\end{equation}
Therefore
\begin{equation}
  -2\varepsilon E + \nabla^2 \log q_T(r_\varepsilon) + B \preccurlyeq 0,
\end{equation}
or
\begin{equation}
  \nabla^2 \log q_T(r_\varepsilon) \preccurlyeq -B + 2\varepsilon E.
\end{equation}
At the same point, Proposition~\ref{prop:separable-curvature} gives a
block-diagonal matrix $K_\varepsilon$ satisfying the opposite bound. Combining
them
\begin{equation}
  -E + \Realif(K_\varepsilon) \preccurlyeq 
  \nabla^2 \log q_T(r_\varepsilon) \preccurlyeq -B + 2\varepsilon E.
\end{equation}
Hence
\begin{equation}\label{eq:EB}
  E - \Realif(K_\varepsilon) \succcurlyeq B - 2\varepsilon E.
\end{equation}
We want to let $\varepsilon\to0$ in Eq.~\eqref{eq:EB}. To do this, we must
first ensure that the matrices $K_\varepsilon$ remain bounded. Let
\begin{equation}
  b = \lambda_{\mathrm{min}}(B) > 0.
\end{equation}
For $0<\varepsilon<b/4$, inequality \eqref{eq:EB} gives
\begin{equation}
  E-\Realif(K_\varepsilon) \succcurlyeq \frac{b}{2}E.
\end{equation}
Set
\begin{equation}
  \delta = \min\left(\frac{b}{2}, \frac{1}{2}\right) > 0.
\end{equation}
Then
\begin{equation}
  \Realif(K_\varepsilon) \preccurlyeq (1-\delta)E.
\end{equation}
Here the special structure of $\mathcal R(K_\varepsilon)$ matters.
Lemma~\ref{lem:realification} shows that its eigenvalues occur in opposite pairs
and that
\begin{equation}
  \|\Realif(K_\varepsilon)\| = \|K_\varepsilon\|.
\end{equation}
Thus the upper bound also controls the negative eigenvalues, giving
\begin{equation}
  \|K_\varepsilon\| \leqslant 1-\delta<1.
\end{equation}
The matrices $K_\varepsilon$ therefore lie in a fixed closed bounded subset of
a finite-dimensional space. Along a sequence $\varepsilon\downarrow0$, we can
choose a convergent subsequence, $K_\varepsilon \to K_*$. The limit remains
complex symmetric and block diagonal across $\mathcal I$, and it satisfies
\begin{equation}\label{eq:KB}
  \|K_*\| < 1, \quad
  E-\Realif(K_*) \succcurlyeq B.
\end{equation}
The uniform gap $\delta$ ensures that the limit still has norm strictly less
than one. The maximizing points $r_\varepsilon$ do not need to converge; only
the matrices do.

We now turn the matrix $K_*$ into a Gaussian state. Lemma~\ref{lem:cayley}
establishes that, because $K_*^{\mathsf T}=K_*$ and $\|K_*\|<1$
\begin{equation}
  \gamma_* = (E-\Realif(K_*))^{-1} - \tfrac{1}{2}E
\end{equation}
is the covariance of a zero-mean pure Gaussian state. Because $K_*$ is block
diagonal across $\mathcal I$, that state is a product across $\mathcal I$.
This is an important role of Lemma~\ref{lem:cayley}: it verifies that
$\gamma_*$ is a physically valid pure-state covariance, not merely a positive
matrix. In particular, it establishes the physicality and purity conditions
\begin{eqnarray}
  \gamma_* + \frac{i}{2}\Omega \succcurlyeq 0, \quad
  \gamma_* \Omega \gamma_* = \tfrac{1}{4}\Omega,
\end{eqnarray}
and checks the product structure. It remains to compare this covariance with
$\gamma$. From Eq.~\eqref{eq:KB}
\begin{equation}
  E-\Realif(K_*) \succcurlyeq B = \left(\gamma + \tfrac{1}{2}E\right)^{-1}.
\end{equation}
Inversion reverses the order of positive-definite matrices. Therefore
\begin{equation}
  (E-\Realif(K_*))^{-1} \preccurlyeq \gamma + \tfrac{1}{2}E.
\end{equation}
Subtracting $E/2$ gives $\gamma_* \preccurlyeq \gamma$.

Let $\rho_*$ be the zero-mean pure product Gaussian state just constructed,
and define
\begin{equation}
  C = \gamma - \gamma_* \succcurlyeq 0.
\end{equation}
Choose a classical Gaussian random displacement $\xi$ with mean $\mu$ and
covariance $C$. For example
\begin{equation}
  \xi = \mu + C^{1/2}Z,
\end{equation}
where $Z$ is a standard real Gaussian vector. This definition also works when
$C$ is singular. Apply the displacement to $\rho_*$, and average:
\begin{equation}\label{eq:rho-tilde}
  \tilde{\rho} = \int D(\xi)\rho_* D^\dagger(\xi) \, d\nu(\xi),
\end{equation}
where $\nu$ is the distribution of $\xi$, and $D(\xi)$ is the displacement
operator. Every displacement factors across the partition:
\begin{equation}
  D(\xi) = D_{I_1}(\xi_1) \otimes \ldots \otimes D_{I_k}(\xi_k).
\end{equation}
Since $\rho_*$ is a product state, each displaced state is also a product state.
Their classical mixture is therefore separable. Note that the displacement
components need not be independent across parties. In particular, $C$ need not
be block-diagonal across $\mathcal{I}$.

Finally, this mixture is exactly $\rho_G$. Its characteristic function is
\begin{equation}
  \chi_{\tilde{\rho}}(t) = e^{-t^{\mathsf T}\gamma_* t/2} e^{it^{\mathsf T}\mu - t^{\mathsf T}C t/2} = 
  e^{it^{\mathsf T}\mu - t^{\mathsf T}\gamma t/2} = \chi_{\rho_G}(t).
\end{equation}
Equality of characteristic functions implies equality of the states,
$\tilde{\rho} = \rho_G$. Thus we have represented $\rho_G$ as a mixture of
product states, proving $\rho_G \in \Sep(\mathcal{I})$. The
characteristic-function calculation is essential: it proves equality of the
complete states, rather than merely equality of their first and second moments.
\end{proof}

It now remains to prove Proposition~\ref{prop:separable-curvature}. We do this
in three steps: first for pure states (where the inequality
\eqref{eq:separable-curvature} becomes matrix equality), then for finite
mixtures and finally for all separable states. The last step is non-trivial ---
according to Ref.~\cite{Holevo2005}, there are separable states that do not have
countable decompositions.

\begin{lemma}\label{lem:product} Let $\ket\psi$ be a normalized vector. At a
point where $q_\psi(r)>0$ there is a complex symmetric matrix $K_\psi(r)$,
$K^T_\psi(r)=K_\psi(r)$, such that
\begin{equation}
    \Hess\log q_\psi(r)=-E+\Realif(K_\psi(r)).
    \label{eq:product-curvature}
\end{equation}
If $\psi$ is an $\mathcal{I}$-product vector, $K_\psi(r)$ is block diagonal
across $\mathcal{I}$.
\end{lemma}

\begin{proof}
The state $\ket\psi$ has the expansion in the Fock basis
\begin{equation}
    \ket\psi = \sum_{m\in\N^n} c_m \ket m.    
\end{equation}
Define 
\begin{equation}
    f_\psi(\alpha) = \sum_{m\in\N^n} \overline{c_m} \frac{\alpha^m}{\sqrt{m!}}.
\end{equation}
We have
\begin{equation}
 \bra\psi \alpha\rangle=e^{-\|\alpha\|^2/2}f_\psi(\alpha),\qquad
 q_\psi(\alpha)=e^{-\|\alpha\|^2}|f_\psi(\alpha)|^2.
\end{equation}
Define a function of the real argument
\begin{equation}
    F(x, p) = f_\psi\left(\frac{x+ip}{\sqrt{2}}\right).
\end{equation}
The Husimi function reads as
\begin{equation}
    q_\psi(x, p) = e^{-(\|x\|^2+\|p\|^2)/2} |F(x, p)|^2.
\end{equation}
At a point $(x, p)$ where $q_\psi(x, p)>0$ we can write
\begin{equation}
    \log q_\psi(x, p) = -\tfrac{1}{2}(\|x\|^2+\|p\|^2) + h(x, p), \quad
    h(x, p) = \log |F(x, p)|^2.
\end{equation}
The Gaussian term has Hessian $-E_{2n}$. We therefore need only calculate the
Hessian of $h$.

It is easy to see that
\begin{equation}
    \frac{\partial F}{\partial p_j} = i \frac{\partial F}{\partial x_j}.
\end{equation}
Differentiating again gives
\begin{equation}
    \frac{\partial^2 F}{\partial x_j \partial p_k} = 
    i\frac{\partial^2 F}{\partial x_j \partial x_k}, \quad
    \frac{\partial^2 F}{\partial p_j \partial p_k} = 
    -\frac{\partial^2 F}{\partial x_j \partial x_k}.
\end{equation}
These identities contain all the analytic structure we need.

For any real coordinate $a$, the product and chain rules give
\begin{equation}
    \frac{\partial h}{\partial a} = 
    \frac{1}{|F|^2} \frac{\partial F F^*}{\partial a} = 
    \frac{\frac{\partial F}{\partial a}F^* + F \frac{\partial F^*}{\partial a}}{|F|^2} = 
    2\re\left(\frac{1}{F}\frac{\partial F}{\partial a}\right).
\end{equation}
Differentiating once more with respect to another real coordinate $b$
\begin{equation}
    \frac{\partial^2 h}{\partial a \partial b} = 
    2\re\left(\frac{1}{F}\frac{\partial^2 F}{\partial a \partial b} - 
    \frac{1}{F^2}\frac{\partial F}{\partial a}\frac{\partial F}{\partial b}\right).
\end{equation}
Now define the $n\times n$ complex matrix
\begin{equation}
    K_{ij} = 2\left(\frac{1}{F}\frac{\partial^2 F}{\partial x_i \partial x_j} - 
    \frac{1}{F^2}\frac{\partial F}{\partial x_i}\frac{\partial F}{\partial x_j}\right).
\end{equation}
Because mixed derivatives commute, the matrix $K_\psi$ is complex symmetric.

For the $xx$-block, equations above give
\begin{equation}
    \frac{\partial^2 h}{\partial x_i \partial x_j} = \re K_{ij}.
\end{equation}
For the mixed block we have
\begin{equation}
    \frac{\partial^2 h}{\partial x_i \partial p_j} = 
    2\re\left(\frac{i}{F}\frac{\partial^2 F}{\partial x_i \partial x_j} - 
    \frac{i}{F^2}\frac{\partial F}{\partial x_i}\frac{\partial F}{\partial x_j}\right) = 
    \re(iK_{ij}) = -\im(K_{ij}).
\end{equation}
For the $pp$-block
\begin{equation}
    \frac{\partial^2 h}{\partial p_i \partial p_j} = 
    2\re\left(-\frac{1}{F}\frac{\partial^2 F}{\partial x_i \partial x_j} + 
    \frac{1}{F^2}\frac{\partial F}{\partial x_i}\frac{\partial F}{\partial x_j}\right) = 
    -\re(K_{ij}).
\end{equation}
Therefore
\begin{equation}
    \nabla^2 h = 
    \begin{pmatrix}
        \re K_\psi & -\im K_\psi \\
        -\im K_\psi & -\re K_\psi
    \end{pmatrix}
    = \Realif(K_\psi).
\end{equation}
Adding the Gaussian contribution, we obtain
\begin{equation}
    \nabla^2 \log q_\psi = -E_{2n} + \Realif(K_\psi).
\end{equation}

Suppose now that $\psi$ is factorizable across a partition $\mathcal I = \{I_1,
\ldots, I_k\}$. The corresponding $F$ then factors:
\begin{equation}
    F(r) = \prod^k_{j=1} F_j(r_{I_j}).
\end{equation}
Where $F\neq0$, every factor is nonzero, and
\begin{equation}
    h(r) = \log |F(r)|^2 = 
    \sum^k_{j=1} \log |F_j(r_{I_j})|^2.
\end{equation}
Each summand depends only on its own party’s coordinates. Hence all second
derivatives coupling different parties vanish. Therefore $K$ is block diagonal
across $\mathcal{I}$.
\end{proof}

We now extend this result for finite mixtures. In general, only inequality is
guaranteed to hold, but that is what we need.
\begin{lemma}\label{prop:finite-product} Let $T=\sum_{\ell=1}^m
w_\ell\proj{\psi_\ell}$, where $w_\ell\geqslant 0$ and every $\psi_\ell$ is a
normalized $\mathcal{I}$-product vector. At every $r$ with $q_T(r)>0$, there is
a complex symmetric $\mathcal{I}$-block-diagonal matrix $K_T(r)$ such that
\begin{equation}
    \Hess\log q_T(r)\succcurlyeq-E+\Realif(K_T(r)).
\end{equation}
\end{lemma}

\begin{proof}
Write $u_\ell(r) = w_\ell q_{\psi_\ell}(r)$, so that 
\begin{equation}
    q_T(r) = \sum^m_{\ell=1}u_\ell(r).
\end{equation}
Fix $r_0$ with $q_T(r_0)>0$ and define the active index set
\begin{equation}
    I_+ = \{\ell: u_\ell(r_0)>0\}.
\end{equation}
For $\ell\in I_+(r_0)$, set
\begin{equation}
    p_\ell = \frac{u_\ell(r_0)}{q_T(r_0)}.
\end{equation}
Then
\begin{equation}
    p_\ell > 0, \quad \sum_{\ell\in I_+} p_\ell = 1.
\end{equation}
These weights are now held fixed as $r$ varies. Since there are finitely many
positive components and they are continuous, all these components remain
positive in a sufficiently small neighborhood of $r_0$.

For $r$ in that neighborhood
\begin{equation}
    \frac{q_T(r)}{q_T(r_0)} \geqslant \sum_{\ell\in I_+} \frac{u_\ell(r)}{q_T(r_0)} = 
    \sum_{\ell\in I_+} p_\ell \frac{u_\ell(r)}{u_\ell(r_0)} \geqslant
    \prod_{\ell\in I_+}\left(\frac{u_\ell(r)}{u_\ell(r_0)}\right)^{p_\ell}.
\end{equation}
The first inequality simply drops nonnegative components. The last is the
weighted arithmetic-geometric mean inequality. Taking logarithms gives
\begin{equation}
    F(r) = \log \frac{q_T(r)}{q_T(r_0)} - 
    \sum_{\ell\in I_+} p_\ell \log\frac{u_\ell(r)}{u_\ell(r_0)} \geqslant 0.
\end{equation}
At $r=r_0$, all ratios equal one, so $F(r_0) = 0$. Therefore $F$ has a local
minimum at $r_0$, and
\begin{equation}
    \nabla^2 F(r_0) \succcurlyeq 0.
\end{equation}
The denominators in the logarithms are constants, and so are the $p_\ell$.
Hence
\begin{equation}\label{eq:qu}
    \nabla^2 \log q_T(r_0) \succcurlyeq \sum_{\ell\in I_+} p_\ell \nabla^2 \log u_\ell(r_0).
\end{equation}
This proves the needed mixture inequality directly. The equality at $r_0$ is
essential: we are using the second-derivative test at a local minimum, not
merely differentiating a pointwise inequality.

Multiplication by constants does not change the logarithmic Hessian:
\begin{equation}
    \nabla^2 \log u_\ell(r) = \nabla^2 \log q_{\psi_\ell}(r).
\end{equation}
Lemma~\ref{lem:product} gives
\begin{equation}
    \nabla^2 \log q_{\psi_\ell}(r) = -E + \Realif(K_\ell),
\end{equation}
where each $K_\ell$ is symmetric and block diagonal across $\mathcal{I}$.
Substituting into \eqref{eq:qu}
\begin{equation}
    \nabla^2 \log q_T(r_0) \succcurlyeq 
    \sum_{\ell\in I_+} p_\ell[-E + \Realif(K_\ell)] = 
    -E + \Realif\left(K_T\right),
\end{equation}
where
\begin{equation}
    K_T = \sum_{\ell\in I_+} p_\ell K_\ell.
\end{equation}
A real weighted average preserves symmetry and block-diagonal structure. This
proves the proposition.
\end{proof}

Now we can prove Proposition~\ref{prop:separable-curvature}.
\begin{proof}[Proof of Proposition~\ref{prop:separable-curvature}]
Fix one point $r$ with $q_T(r)>0$. All the following matrices and
derivatives are evaluated at this fixed point. Since $T$ is a nonzero
separable positive operator, write
\begin{equation}
    T = t\sigma, \quad t = \Tr T > 0, \quad \sigma \in \Sep(\mathcal{I}).
\end{equation}
Every local density operator is a trace-norm limit of finite convex combinations
of pure states by its spectral decomposition; hence every product density
operator, and therefore every finite mixture of product density operators, can
be approximated in trace norm by finite mixtures of pure product
projectors. Thus there are finite positive mixtures
\begin{equation}
    T_j = \sum^{m_j}_{\ell=1} w_{j\ell}\proj{\psi_{j\ell}},  \quad w_{j\ell} \geqslant 0,
\end{equation}
where every $\psi_{j\ell}$ is a normalized product vector across
$\mathcal{I}$, such that
\begin{equation}
    \|T_j-T\|_1 \to 0.
\end{equation}
We need only this approximation; we do not assume that $T$ itself has an exact
finite or countable product-state decomposition.

Trace-norm convergence of the operators implies convergence of their Husimi
functions and their first two derivatives. In particular $q_{T_j}(r) \to q_T(r)
> 0$. Thus, for sufficiently large $j$, $q_{T_j}(r)>0$, and its logarithmic
Hessian is well defined. Write
\begin{equation}
    H_j = \nabla^2 \log q_{T_j}(r), \quad
    H = \nabla^2 \log q_{T}(r).
\end{equation}
Then $H_j \to H$. To see why, use
\begin{equation}
    \nabla^2 \log q = \frac{\nabla^2 q}{q} - \frac{\nabla q (\nabla q)^{\mathsf T}}{q^2}.
\end{equation}
The values and derivatives in this expression converge, and the denominators
converge to the strictly positive number $q_T(r)$.

Since each $T_j$ is a finite product mixture, Lemma~\ref{prop:finite-product}
supplies a complex symmetric, $\mathcal{I}$-block-diagonal matrix $K_j$
satisfying $H_j \succcurlyeq -E + \Realif(K_j)$. Equivalently
\begin{equation}\label{eq:KH}
    \Realif(K_j) \preccurlyeq E + H_j.
\end{equation}
We would now like to pass to the limit in this inequality. Before doing so, we
must show that the $K_j$ have a convergent subsequence.

Because $H_j\to H$, the matrices $I+H_j$ are bounded. Therefore there is a
single constant $c>0$, independent of $j$, such that $E+H_j \preccurlyeq c
E$ for all sufficiently large $j$. Combining this with Eq.~\eqref{eq:KH} we
derive
\begin{equation}\label{eq:RE}
    \Realif(K_j) \preccurlyeq cE.
\end{equation}
For an arbitrary symmetric matrix, an upper bound like this would not control
the operator norm: its negative eigenvalues could still become arbitrarily large
in magnitude. But the matrices $\mathcal R(K_j)$ have a special symmetry.
Lemma~\ref{lem:realification} shows that their eigenvalues occur in opposite
pairs: whenever $\lambda$ occurs, so does $-\lambda$. Consequently,
Eq.~\eqref{eq:RE} bounds both $|\lambda|\leqslant c$. Thus every eigenvalue lies
in $[-c,c]$, and
\begin{equation}
    -cE \preccurlyeq \Realif(K_j) \preccurlyeq cE.
\end{equation}
The same lemma gives
\begin{equation}
    \|\Realif(K_j)\| = \|K_j\|.
\end{equation}
Hence $\|K_j\| \leqslant c$. Each $K_j$ is an $n\times n$ matrix. It therefore
has only finitely many real and imaginary entries, and this inequality bounds
all of them. We can consequently choose a subsequence along which every entry
converges, $K_{j_k} \to K_T$. The limit retains the required structure. Indeed,
$(K_{j_k})_{ab} = (K_{j_k})_{ba}$ passes to the limit, so $K^{\mathsf T}_T=K_T$.
Likewise, every entry connecting different blocks of $\mathcal{I}$ was zero and
remains zero. Thus $K_T$ is still complex symmetric and $\mathcal{I}$-block
diagonal.

Finally, for every real vector $v$, inequality \eqref{eq:KH} says
\begin{equation}
    v^{\mathsf T}\Realif(K_{j_k})v \leqslant v^{\mathsf T}(E+H_{j_k})v.
\end{equation}
Taking the limit gives
\begin{equation}
    v^{\mathsf T}\Realif(K_T)v \leqslant v^{\mathsf T}(E+H)v.
\end{equation}
Because this holds for every $v$, we have $\Realif(K_T) \preccurlyeq E+H$.
Recalling the definition of $H$, we obtain
\begin{equation}
    \nabla^2 \log q_T(r) \succcurlyeq -E +\Realif(K_T).
\end{equation}
This proves the proposition.
\end{proof}

\section*{Acknowledgements}

\noindent The authors thank the BMFTR (via QR.N, QuaPhySI, and QuKuK) and the EU
(via CLUSTEC) for their support.

\appendix

\section{Construction 1}

The construction modifies $\tau_3$ only on the eight-dimensional subspace
\begin{equation}
\mathcal H_{01}=\operatorname{span}\{\ket{abc}:a,b,c\in\{0,1\}\}.
\end{equation}
Let $\mathcal{P}_{01}$ denote the orthogonal projection onto this subspace
\begin{equation}
    \mathcal{P}_{01} = (\proj{0}+\proj{1})^{\otimes 3}.
\end{equation}
The total weight of $\tau_3$ on $\mathcal H_{01}$ is
\begin{equation}
 p_{01}=\Tr(\mathcal{P}_{01}\tau_3)=\left(\frac34\right)^3=\frac{27}{64}.
\label{eq:t}
\end{equation}
The normalized restriction of $\tau_3$ to $\mathcal H_{01}$ is
\begin{equation}
D=\frac{1}{p_{01}}\mathcal{P}_{01}\tau_3\mathcal{P}_{01}
=\frac1{27}\operatorname{diag}(8,4,4,2,4,2,2,1),
\label{eq:D}
\end{equation}
where $D$ is understood to be extended by zero on $\mathcal H_{01}^{\perp}$, and throughout this section the ordered basis of $\mathcal H_{01}$ is
\begin{equation}
\ket{000},\ket{001},\ket{010},\ket{011},
\ket{100},\ket{101},\ket{110},\ket{111}.
\label{eq:basis}
\end{equation}
Thus
\begin{equation}
\tau_3=R+p_{01}D, \qquad
R=\tau_3-p_{01}D=\tau_3-\mathcal{P}_{01}\tau_3\mathcal{P}_{01}.
\label{eq:Rdef-1}
\end{equation}
The operator $R$ is positive, diagonal in the product Fock basis, fully separable, and
\begin{equation}
\Tr R=1-p_{01}=\frac{37}{64}.
\end{equation}

We next construct a normalized state $\rho_1$ supported on $\mathcal H_{01}$
which is separable across $1|23$ but NPT across the other two cuts.  Define
\begin{equation}
\ket{\pm}=\frac{\ket0\pm\ket1}{\sqrt2},
\qquad
\ket{\pm i}=\frac{\ket0\pm i\ket1}{\sqrt2},
\qquad
\ket{\Phi^+}_{23}=\frac{\ket{00}_{23}+\ket{11}_{23}}{\sqrt2}.
\end{equation}
For compactness, $\ket{u;vw}$ means $\ket u_1\otimes\ket v_2\otimes\ket w_3$, while
$\ket{0;\Phi^+}$ means $\ket0_1\otimes\ket{\Phi^+}_{23}$.  Put
\begin{equation}
\begin{split}
\rho_1 &=
\frac{2}{27}\proj{0;00}+\frac{2}{9}\proj{0;\Phi^+}
+\frac{2}{9}\proj{1;00}+\frac{1}{27}\proj{1;11}\\
&+\frac1{18}\Bigl(\proj{+;0-}+\proj{+;-0}+\proj{-;0+}+\proj{-;+0}\Bigr)\\
&+\frac1{18}\Bigl(\proj{+i;0,+i}+\proj{+i;+i,0}\Bigr)\\
&+\frac1{18}\Bigl(\proj{-i;0,-i}+\proj{-i;-i,0}\Bigr).
\label{eq:rho1decomp}
\end{split}
\end{equation}
The coefficients in \eqref{eq:rho1decomp} are nonnegative and sum to one.  More importantly, every pure state appearing in \eqref{eq:rho1decomp} is a product vector across the cut $1|23$.  Hence
\begin{equation}
\rho_1\in\Sep(1|23).
\label{eq:rho1sep}
\end{equation}
Notice that $\ket{0;\Phi^+}$ need not be a three-mode product vector: it is enough that it is a product across the indicated bipartition.

In the basis \eqref{eq:basis}, direct expansion of \eqref{eq:rho1decomp} gives
\begin{equation}
\rho_1=\frac1{54}
\begin{pmatrix}
16&0&0&6&0&-3&-3&0\\
0&3&0&0&0&0&0&0\\
0&0&3&0&0&0&0&0\\
6&0&0&6&0&0&0&0\\
0&0&0&0&18&0&0&0\\
-3&0&0&0&0&3&0&0\\
-3&0&0&0&0&0&3&0\\
0&0&0&0&0&0&0&2
\end{pmatrix}.
\label{eq:rho1matrix}
\end{equation}

Let $S$ be the unitary implementing the cyclic shift
\begin{equation}
S\ket{a,b,c}=\ket{c,a,b}.
\label{eq:cyclic}
\end{equation}
Define
\begin{equation}
\rho_2=S\rho_1S^{\dagger},
\qquad
\rho_3=S^2\rho_1S^{\dagger 2}.
\label{eq:rho23}
\end{equation}
The original mode $1$ is sent by $S$ to mode $2$, and by $S^2$ to mode $3$.  Therefore
\begin{equation}
\rho_2\in\Sep(2|13),
\qquad
\rho_3\in\Sep(3|12).
\label{eq:rho23sep}
\end{equation}

The reason for the particular coefficients in \eqref{eq:rho1decomp} is the exact cyclic cancellation
\begin{equation}
\frac{\rho_1+\rho_2+\rho_3}{3}=D.
\label{eq:cyclicaverage}
\end{equation}
This can be checked immediately from \eqref{eq:rho1matrix}.  For example, the three coherences
\[
\ket{000}\!\bra{011},\qquad
\ket{000}\!\bra{101},\qquad
\ket{000}\!\bra{110}
\]
appear in $\rho_1$ with coefficients
\[
\frac19,\qquad -\frac1{18},\qquad -\frac1{18},
\]
respectively.  Cyclic permutation rotates these three coefficients, so each off-diagonal matrix element vanishes in the average because
\[
\frac19-\frac1{18}-\frac1{18}=0.
\]
The diagonal entries of the same average are
\[
\frac1{27}(8,4,4,2,4,2,2,1),
\]
which proves \eqref{eq:cyclicaverage}.

We now reattach the common thermal remainder $R$.  Define
\begin{equation}
\sigma_j=R+p_{01}\rho_j,
\qquad j=1,2,3,
\label{eq:sigmadef}
\end{equation}
with $p_{01}=27/64$. Since $\Tr R=37/64$ and $\Tr\rho_j=1$, every $\sigma_j$ is
a normalized state. Moreover, $R$ is fully separable, so \eqref{eq:rho1sep} and
\eqref{eq:rho23sep} imply
\begin{equation}
\sigma_1\in\Sep(1|23),
\qquad
\sigma_2\in\Sep(2|13),
\qquad
\sigma_3\in\Sep(3|12).
\label{eq:sigmasep}
\end{equation}
Using \eqref{eq:Rdef-1} and \eqref{eq:cyclicaverage},
\begin{equation}
\frac{\sigma_1+\sigma_2+\sigma_3}{3}
=R+\frac{p_{01}}{3}(\rho_1+\rho_2+\rho_3)
=R+p_{01}D =\tau_3.
\label{eq:mainaverage}
\end{equation}
Thus the Gaussian state $\tau_3$ is an equal three-term mixture whose three
terms are separable across three different bipartitions.

The construction is stronger than merely assigning a different separable cut to
each component.  Each $\sigma_j$ is NPT across its other two bipartitions.
Consider first $\sigma_1$ and partial transposition on mode $2$.  The matrix
element
\begin{equation}
\frac19\ket{000}\!\bra{011}+\frac19\ket{011}\!\bra{000}
\end{equation}
in $\rho_1$ becomes
\begin{equation}
\frac19\ket{010}\!\bra{001}+\frac19\ket{001}\!\bra{010}
\end{equation}
after the partial transposition of the second mode $\mathsf{T}_2$.  Since $R$
has zero support on $\mathcal H_{01}$, the principal block of
$\sigma_1^{\mathsf{T}_2}$ on $\operatorname{span}\{\ket{001},\ket{010}\}$ is
exactly
\begin{equation}
\frac{27}{64}
\begin{pmatrix}
1/18&1/9\\
1/9&1/18
\end{pmatrix}
=
\begin{pmatrix}
3/128&3/64\\
3/64&3/128
\end{pmatrix}.
\label{eq:NPTblock}
\end{equation}
Its eigenvalues are $9/128$ and $-3/128$.  Hence
\begin{equation}
\sigma_1^{\mathsf{T}_2}\not\succcurlyeq 0,
\end{equation}
so $\sigma_1$ is entangled across $2|13$.  The state $\rho_1$, and hence
$\sigma_1$, is invariant under the exchange of modes $2$ and $3$, so the same
argument gives
\begin{equation}
\sigma_1^{\mathsf{T}_3}\not\succcurlyeq 0.
\end{equation}
Thus $\sigma_1$ is NPT across both cuts other than $1|23$.  By cyclic
permutation, the complete pattern is
\begin{equation}
\begin{array}{c|ccc}
&1|23&2|13&3|12\\
\hline
\sigma_1&\text{separable}&\text{NPT}&\text{NPT}\\
\sigma_2&\text{NPT}&\text{separable}&\text{NPT}\\
\sigma_3&\text{NPT}&\text{NPT}&\text{separable}.
\end{array}
\label{eq:cuttable}
\end{equation}

All three states $\sigma_j$ are centered.  Indeed, the difference $\rho_1-D$
contains only diagonal terms and coherences connecting Fock vectors which differ
in two modes, so every first moment vanishes; cyclic permutation gives the same
statement for $\rho_2-D$ and $\rho_3-D$.  Since $\tau_3$ is centered, so are the
$\sigma_j$.

For $\sigma_1$, a direct calculation gives, in the ordering
$r=(x_1,x_2,x_3,p_1,p_2,p_3)^T$,
\begin{equation}
\gamma_{\sigma_1}=\frac32E_6+\frac1{128}
\begin{pmatrix}
 8&-3&-3&0&0&0\\
-3&-4& 6&0&0&0\\
-3& 6&-4&0&0&0\\
0&0&0& 8& 3& 3\\
0&0&0& 3&-4&-6\\
0&0&0& 3&-6&-4
\end{pmatrix}.
\label{eq:gamma1compact}
\end{equation}
Let $S_{\rm ph}$ denote the corresponding cyclic permutation matrix on phase
space. Then
\begin{equation}
\gamma_{\sigma_2}=S_{\rm ph}\gamma_{\sigma_1}S_{\rm ph}^{\mathsf{T}},
\qquad
\gamma_{\sigma_3}=S_{\rm ph}^2\gamma_{\sigma_1}(S_{\rm ph}^{\mathsf{T}})^2,
\label{eq:gamma23}
\end{equation}
where $S_{\rm ph}$ is the $6\times6$ real permutation matrix implementing on
phase space the same cyclic permutation as $S$
\begin{equation}
    S_{\rm ph} = 
    \begin{pmatrix}
        P & 0 \\
        0 & P
    \end{pmatrix},
    \quad
    P = 
    \begin{pmatrix}
        0 & 0 & 1 \\
        1 & 0 & 0 \\
        0 & 1 & 0
    \end{pmatrix}.
\end{equation}
In particular
\begin{equation}
\gamma_{\sigma_1}\neq\gamma_{\tau_3},
\qquad
\gamma_{\sigma_2}\neq\gamma_{\tau_3},
\qquad
\gamma_{\sigma_3}\neq\gamma_{\tau_3},
\end{equation}
and the three component covariances are mutually different.  Since the three
states have the same zero mean and their average is $\tau_3$, covariance is
affine here, and therefore
\begin{equation}
\frac{\gamma_{\sigma_1}+\gamma_{\sigma_2}+\gamma_{\sigma_3}}{3}
=\gamma_{\tau_3}=\frac32E_6.
\label{eq:covavg}
\end{equation}

\section{Construction 2}

Define two positive operators by
\begin{equation}
\begin{split}
A_+ &=\frac1{64}\Big(6\proj{000}+8\proj{001}+8\proj{010}+4\proj{011}\\
&+\proj{111}+\ket{000}\!\bra{111}+\ket{111}\!\bra{000}\Big),\\
A_- &=\frac1{64}\Big(10\proj{000}+8\proj{100}+4\proj{101}+4\proj{110}\\
&+\proj{111}-\ket{000}\!\bra{111}-\ket{111}\!\bra{000}\Big).
\label{eq:Aminus}
\end{split}
\end{equation}
Their positivity is immediate except possibly on the span of $\ket{000}$ and
$\ket{111}$.  On that subspace the two matrices are, respectively,
\begin{equation}
\frac1{64}
\begin{pmatrix}
6&1\\
1&1
\end{pmatrix},
\qquad
\frac1{64}
\begin{pmatrix}
10&-1\\
-1&1
\end{pmatrix}.
\label{eq:2x2blocks}
\end{equation}
Their determinants are $5/64^2$ and $9/64^2$, and their traces are positive.  Hence
\begin{equation}
A_+\succcurlyeq 0,
\quad
A_-\succcurlyeq 0.
\label{eq:Apositive}
\end{equation}
Furthermore,
\begin{equation}
\Tr A_+=\Tr A_-=\frac{27}{64}.
\label{eq:Atrace}
\end{equation}

The key cancellation is
\begin{equation}
\frac{A_++A_-}{2}=p_{01}D.
\label{eq:Aaverage}
\end{equation}
Indeed, the opposite $\ket{000}\!\bra{111}$ coherences cancel, while the
averaged diagonal is
\[
\frac1{64}(8,4,4,2,4,2,2,1),
\]
which is exactly $p_{01}D$.

We now define two normalized states on the full three-mode Hilbert space:
\begin{equation}
\sigma_+=R+A_+,
\qquad
\sigma_-=R+A_-.
\label{eq:sigmadef2}
\end{equation}
We have
\[
\Tr\sigma_\pm=\frac{37}{64}+\frac{27}{64}=1,
\]
so $\sigma_\pm$ are density operators.  Equation~\eqref{eq:Aaverage} gives the exact equal-mixture identity
\begin{equation}
\tau_3=\frac12\sigma_++\frac12\sigma_-.
\label{eq:mainmixture}
\end{equation}
Thus a fully separable Gaussian state is the equal mixture of two states. It
remains to prove that both of those states are GME.

We first recall a simple matrix-element inequality for three-qubit states.  For
any positive operator $\omega$ supported on $(\mathbb C^2)^{\otimes3}$, write
\[
\omega_{abc,a'b'c'}=\bra{abc}\omega\ket{a'b'c'}.
\]
According to Ref.~\cite{GuhneSeevinck2010}, every biseparable positive operator $\omega$ satisfies
\begin{equation}
|\omega_{000,111}|
\leqslant
\sqrt{\omega_{011,011}\omega_{100,100}}
+
\sqrt{\omega_{101,101}\omega_{010,010}}
+
\sqrt{\omega_{110,110}\omega_{001,001}}.
\label{eq:GMEwitness}
\end{equation}
For $A_+$ we have
\begin{equation}
|(A_+)_{000,111}|=\frac1{64}>0,
\label{eq:Apluscoherence}
\end{equation}
while
\begin{equation}
(A_+)_{100,100}=(A_+)_{101,101}=(A_+)_{110,110}=0.
\label{eq:Apluszeros}
\end{equation}
Every term on the right-hand side of Eq.~\eqref{eq:GMEwitness} therefore vanishes.  Hence $A_+$ violates the biseparability inequality and is GME after normalization.

Likewise,
\begin{equation}
|(A_-)_{000,111}|=\frac1{64}>0,
\label{eq:Aminuscoherence}
\end{equation}
whereas
\begin{equation}
(A_-)_{001,001}=(A_-)_{010,010}=(A_-)_{011,011}=0.
\label{eq:Aminuszeros}
\end{equation}
Again the right-hand side of Eq.~\eqref{eq:GMEwitness} is zero, so the normalized version of $A_-$ is GME.

We now lift this conclusion from the finite-dimensional sector to the full states $\sigma_\pm$.  The operation
\begin{equation}
\rho\longmapsto \mathcal{P}_{01}\rho\mathcal{P}_{01}
\label{eq:localfilter}
\end{equation}
is a local filter.  Local filtering cannot create GME from a biseparable state:
a term separable across a given bipartition remains separable across the same
bipartition after local filtering, and a convex mixture of such filtered terms
remains biseparable.  But Eq.~\eqref{eq:sigmadef2} give
\begin{equation}
\mathcal{P}_{01}\sigma_\pm\mathcal{P}_{01} = A_\pm.
\label{eq:filtered}
\end{equation}
Conditioned on success, the filtered state is $A_\pm/\Tr A_\pm$, which is GME by the argument above.  Therefore the original states cannot be biseparable:
\begin{equation}
\sigma_+\ \text{and}\ \sigma_-\ \text{are both genuinely multipartite entangled.}
\label{eq:GMEconclusion}
\end{equation}
Combining this with Eq.~\eqref{eq:mainmixture} proves the announced decomposition.

\section{Construction 3}
\label{app:C3}

Here we construct a three-mode non-Gaussian state that has three entangled
Gaussian two-mode marginals. Let $\rho_0$ be centered with covariance
$\gamma_0=X_0\oplus Y_0$, in the ordering $(x_1,x_2,x_3,p_1,p_2,p_3)$, where
\begin{equation}
 X_0=\frac14\begin{pmatrix}2&1&1\\1&2&1\\1&1&2\end{pmatrix},
 \qquad
 Y_0=\frac14\begin{pmatrix}3&-1&-1\\-1&3&-1\\-1&-1&3\end{pmatrix}.
 \label{eq:initial-covariance}
\end{equation}
Both matrices are positive definite, and direct multiplication gives
\begin{equation}
 X_0Y_0=\frac14E_3.
 \label{eq:initial-purity}
\end{equation}
Thus $\gamma_0$ is a physical pure Gaussian covariance. For example, with
$\Omega=\left(\begin{smallmatrix}0&E_3\\-E_3&0\end{smallmatrix}\right)$,
one has $\gamma_0\Omega\gamma_0=\Omega/4$.

Set $\delta=1/8$ and let
\begin{equation}
 g_\delta(x)=\frac{1}{\sqrt{2\pi\delta}}
             e^{-x^2/(2\delta)}.
 \label{eq:noise-density-one}
\end{equation}
Define a classical probability density on $\mathbb R^3$ by
\begin{equation}
 p(\xi_1,\xi_2,\xi_3)=
 \prod_{j=1}^3g_\delta(\xi_j)
 \left[1+\frac12\prod_{j=1}^3
        \sin\left(\frac{\xi_j}{\sqrt\delta}\right)\right].
 \label{eq:noise-density-three}
\end{equation}
The bracket lies between $1/2$ and $3/2$, so the density is nonnegative.
It integrates to one because each sine has zero Gaussian average.
Integrating out any one coordinate eliminates the sine-product term:
\begin{equation}
 \int_{\mathbb R}p(\xi_1,\xi_2,\xi_3)\,d\xi_3
 =g_\delta(\xi_1)g_\delta(\xi_2),
 \label{eq:noise-pair-marginal}
\end{equation}
and similarly for the other two pairs. Every pair of noise variables is
therefore jointly Gaussian, although their full joint distribution need
not be Gaussian.

Apply the local position displacements
\begin{equation}
 U(\xi)=\exp\left(-i\sum_{j=1}^3\xi_jp_j\right),
 \qquad
 \rho=\int_{\mathbb R^3}p(\xi)\,
          U(\xi)\rho_0U(\xi)^\dagger\,d^3\xi.
 \label{eq:non-gaussian-state}
\end{equation}
Here $U(\xi)^\dagger x_jU(\xi)=x_j+\xi_j$ and the momentum quadratures
are unchanged. Positivity and normalization follow directly from the
classical mixture of displaced states.

On tracing out one mode, the unitary on that mode drops out of the partial
trace. Equation~\eqref{eq:noise-pair-marginal} then shows that the
remaining two modes undergo Gaussian displacement noise. Hence all three
two-mode marginals of $\rho$ are Gaussian. They are identical up to
relabelling, with covariance
\begin{equation}
 \gamma_{ij}=
 \begin{pmatrix}
 5/8&1/4&0&0\\
 1/4&5/8&0&0\\
 0&0&3/4&-1/4\\
 0&0&-1/4&3/4
 \end{pmatrix}
 \quad\text{in the ordering }(x_i,x_j,p_i,p_j).
 \label{eq:pair-covariance}
\end{equation}
The added position variance is precisely $\delta=1/8$.

Partial transposition on the second mode changes the off-diagonal
momentum entries to $+1/4$. The position and transposed momentum blocks
commute. Multiplying the corresponding eigenvalues in their common
eigenbasis gives the squared symplectic eigenvalues
\begin{equation}
 \widetilde\nu_1^2=\frac3{16},\qquad
 \widetilde\nu_2^2=\frac78.
 \label{eq:pair-PT-eigenvalues}
\end{equation}
Thus $\widetilde\nu_{\min}=\sqrt3/4<1/2$: every pair is NPT and therefore
entangled. Proposition~\ref{prop:connected} proves that $\rho$ is GME; in fact,
any two of its three pair marginals suffice.

All first moments vanish. The centered Gaussian input has vanishing third
position moments, and the single-coordinate means of the noise also
vanish. Expanding the displaced moment therefore gives
\begin{equation}
 \langle x_1x_2x_3\rangle_\rho
 =\mathbb E[\xi_1\xi_2\xi_3].
 \label{eq:third-moment-reduction}
\end{equation}
For a centered real Gaussian variable $Z$ of variance $\delta$,
\begin{equation}
 \mathbb E\!\left[Z\sin\left(Z/\sqrt\delta\right)\right]
 =\sqrt\delta\,e^{-1/2},
 \label{eq:sine-moment}
\end{equation}
as follows by differentiating its characteristic function
$\mathbb E[e^{itZ}]=e^{-\delta t^2/2}$. Equation~\eqref{eq:noise-density-three}
then yields
\begin{equation}
 \langle x_1x_2x_3\rangle_\rho
 =\frac12\delta^{3/2}e^{-3/2}
 =\frac{\sqrt2}{64}e^{-3/2}\ne0.
 \label{eq:nonzero-third-moment}
\end{equation}
The commuting position quadratures of a centered Gaussian state have a centered
multivariate Gaussian distribution and hence zero third moments, so $\rho$
defined by Eq.~\eqref{eq:non-gaussian-state} is not Gaussian. Its GME is
nevertheless certified using only its Gaussian pair marginals.

\section{Technical results}
\label{app:D}
Here we prove some technical results which are needed for completeness of the
main part.

We first prove that trace-norm convergence of the operators implies convergence
of their Husimi functions and their first two derivatives. In fact, let
\begin{equation}
    P(r) = \ket{\alpha(r)}\bra{\alpha(r)},
\end{equation}
where $\alpha(r) = (x+ip)/\sqrt{2}$ for $r=(x, p)$. The map $r\mapsto P(r)$ is
$C^\infty$ in operator norm. Therefore, for every real multi-index $\beta$
with $|\beta|\leqslant 2$
\begin{equation}
    \frac{\partial^{|\beta|} q_T}{\partial r^\beta}(r) = 
    \Tr\left[T \frac{\partial^{|\beta|} P}{\partial r^\beta}(r)\right].
\end{equation}
Consequently
\begin{equation}
    \left|\frac{\partial^{|\beta|} q_{T_j}}{\partial r^\beta}(r) - 
    \frac{\partial^{|\beta|} q_T}{\partial r^\beta}(r)\right| \leqslant
    \|T_j - T\|_1 \left\|\frac{\partial^{|\beta|} P}{\partial r^\beta}(r)\right\|,
\end{equation}
where we use the inequality
\begin{equation}
    |\Tr(AB)| \leqslant \|A\|_1 \|B\|,
\end{equation}
so from the convergence $T_j \to T$ in the norm $\|\cdot\|_1$ follows
convergence of all partial derivatives. The norm $\|\cdot\| \equiv
\|\cdot\|_{\infty}$ is the standard operator norm given by
Eq.~\eqref{eq:st-norm}.

\begin{lemma}
In the definition of $\Sep(\mathfrak{I})$ we can omit the closure operation since
\begin{equation}
 \overline{\conv}\biggl(\bigcup_{\mathcal{I}\in\mathfrak{I}}\Sep(\mathcal{I})\biggr) = 
 \conv\biggl(\bigcup_{\mathcal{I}\in\mathfrak{I}}\Sep(\mathcal{I})\biggr).
\end{equation}    
\end{lemma}
\begin{proof}
Let $\rho_j \to \rho$ in trace norm, where each $\rho_j$ is in the algebraic
convex hull of the fixed-partition separable sets
\begin{equation}
    \rho_j = \sum_{\mathcal{I}\in\mathfrak{I}} T_{\mathcal{I}, j} \equiv 
    \sum_{\mathcal{I}\in\mathfrak{I}} p_{\mathcal{I}, j} \rho_{\mathcal{I}, j},
\end{equation}
where $p_{\mathcal{I}, j}$ are nonnegative and sum up to 1 (over $\mathcal{I}$
for any fixed $j$) and $\rho_{\mathcal{I}, j} \in \Sep(\mathcal{I})$.

Choose finite-rank projections $P_N\uparrow E$. Trace-norm convergence implies
uniform trace tightness:
\begin{equation}
    \varepsilon_N \equiv \sup_j\Tr[(E-P_N)\rho_j] \to 0.
\end{equation}
To verify this, first control all sufficiently large $j$ by their trace-norm
distance from $\rho$, then choose $N$ large enough to control $\rho$ and
the finitely many remaining $\rho_j$. Positivity gives
\begin{equation}
    \Tr[(E-P_N)T_{\mathcal{I}, j}] \leqslant \varepsilon_N.
\end{equation}
For any positive trace-class $T$ and orthogonal projection $P$
\begin{equation}
    \|T - PTP\|_1 \leqslant 2 \sqrt{\Tr T \Tr[(E-P)T]}.
\end{equation}
This follows by writing
\begin{equation}
    T - PTP = (E-P)T + PT(E-P)
\end{equation}
and factoring through $\sqrt{T}$, followed by the Hilbert-Schmidt
Cauchy-Schwarz inequality. Consequently
\begin{equation}
    \|T_{\mathcal{I}, j}-P_N T_{\mathcal{I}, j} P_N\|_1 \leqslant 2\sqrt{\varepsilon_N}.
\end{equation}
For fixed $N$, the compressed operators lie in a bounded subset of a
finite-dimensional space. The uniform approximation above therefore makes each
family $\{T_{\mathcal{I},j}\}_j$ relatively compact in trace norm.

Because $\mathfrak{I}$ is finite, take a common subsequence along which
$T_{\mathcal{I}, j} \to T_{\mathcal{I}}$ for every $\mathcal{I}$. Then
\begin{equation}
    T_{\mathcal{I}} \succcurlyeq 0, \quad 
    \rho = \sum_{\mathcal{I}\in\mathfrak{I}} T_{\mathcal{I}}.
\end{equation}
Put $p_\mathcal{I}=\operatorname{Tr}T_\mathcal{I}$. If $p_\mathcal{I}>0$,
normalization and trace-norm closedness of $\operatorname{Sep}(\mathcal{I})$
give
\begin{equation}
    \frac{T_{\mathcal{I}}}{p_{\mathcal{I}}} \in \Sep(\mathcal{I}).
\end{equation}
If $p_\mathcal{I}=0$, positivity implies $T_\mathcal{I}=0$. Thus
\begin{equation}
    \rho = \sum_{\mathcal{I}\in\mathfrak{I}} p_\mathcal{I} \rho_\mathcal{I}, \quad
    \sum_{\mathcal{I}\in\mathfrak{I}} p_\mathcal{I} = 1, \quad
    \rho_\mathcal{I} \in \Sep(\mathcal{I}),
\end{equation}
which completes the proof.
\end{proof}

In the main text we used the construction 
\begin{equation}
    \xi = \mu + \sqrt{C}Z, \quad
    C = \gamma - \gamma_*\succcurlyeq 0,
\end{equation}
and the integral over the probability distribution $d\nu(\xi)$ of the random
vector $\xi$. It is not required that distribution to have a density. Here
\begin{equation}
    Z = (Z_1, \ldots, Z_{2n})^{\mathsf T},
\end{equation}
where the components $Z_j$ are iid random variables with distribution
$\mathcal{N}(0, 1)$. Thus
\begin{equation}
    \mathbb{E}[Z] = 0, \quad
    \mathbb{E}[ZZ^{\mathsf T}] = E.
\end{equation}
The mean and covariance of $\xi$ follow directly
\begin{equation}
\begin{split}
    \mathbb{E}[\xi] &= \mu+\sqrt{C}\mathbb{E}[Z] = \mu \\
    \mathrm{Cov}(\xi) &= \mathbb{E}[(\xi-\mu)(\xi-\mu)^{\mathsf T}] = 
    \sqrt{C}\mathbb{E}[ZZ^{\mathsf T}]\sqrt{C}^{\mathsf T} = C.
\end{split}
\end{equation}
An affine transformation of independent standard normal
variables is the general construction of a multivariate Gaussian random vector,
including a degenerate one. Notice that $Z$ and $\xi$ are classical random
vectors. In particular, $\sqrt{C}$ is acting on ordinary random numbers, not
implementing a transformation of quantum quadrature operators.

Let $g$ denote standard Gaussian probability measure on $\R^{2n}$:
\begin{equation}
    dg(z) = \frac{1}{(2\pi)^n} e^{-\|z\|^2/2}dz.
\end{equation}
Then for any bounded measurable scalar function F
\begin{equation}
    \int_{\R^{2n}} F(\xi)\, d\nu(\xi) = 
    \mathbb{E}[F(\xi)] = 
    \frac{1}{(2\pi)^n} \int_{\R^{2n}} F\left(\mu+\sqrt{C}z\right)e^{-\|z\|^2/2} \, dz.
\end{equation}
This identity is the operational meaning of $d\nu(\xi)$: average $F$ over
the displacements produced by sampling $Z$. There is no inverse or determinant
of $C$ in this formula, so it works unchanged when $C$ is singular. It is
not an invertible change-of-variables argument; it is the definition of the
distribution induced by a random variable.

If $C$ is positive-definite and thus non-degenerate, the measure $d\nu(\xi)$ has
a density
\begin{equation}
    d\nu(\xi) = \frac{1}{(2\pi)^n\sqrt{\det C}}
    \exp\left[-\tfrac{1}{2}(\xi-\mu)^{\mathsf T}C^{-1}(\xi-\mu)\right]d\xi.
\end{equation}
The measure integral becomes an ordinary integral weighted by this Gaussian
density. The density formula requires positive definiteness, whereas the affine
construction above does not. At the extreme $C=0$, the displacement is
deterministic:
\begin{equation}
    \xi = \mu, \quad \nu = \delta_\mu, \quad
    \int_{\R^{2n}} F(\xi) \, d\nu(\xi) = F(\mu).
\end{equation}

In the main part we define
\begin{equation}
  \tilde{\rho} = \int D(\xi)\rho_* D^\dagger(\xi) \, d\nu(\xi),
\end{equation}
where $\nu$ is precisely the probability measure just constructed and
\begin{equation}
    D(\xi) = e^{-i\xi^{\mathsf T}\Omega r}.
\end{equation}
These operators displace the operators $r$
\begin{equation}
    D(\xi)^\dagger r D(\xi) = r + \xi.
\end{equation}
Put
\begin{equation}
    \sigma(\xi) = D(\xi)\rho_* D^\dagger(\xi).
\end{equation}
For every realized displacement, $\sigma(\xi)$ is a density operator. The
integral is its classical average $\tilde{\rho} = \mathbb{E}[\sigma(\xi)]$. We
can now write Eq.~\eqref{eq:rho-tilde} explicitly as
\begin{equation}
    \tilde{\rho} = \frac{1}{(2\pi)^n}
    \int_{\R^{2n}} D\left(\mu+\sqrt{C}z\right)\rho_*D\left(\mu+\sqrt{C}z\right)^\dagger
    e^{-\|z\|^2/2} \, dz.
\end{equation}
This formula is valid for both singular and nonsingular $C$.

Fixed displacement gives
\begin{equation}
    \chi_{\sigma(\xi)}(t) = e^{it^{\mathsf T}\xi} \chi_{\rho_*}(t).
\end{equation}
Averaging yields
\begin{equation}
    \chi_{\tilde{\rho}}(t) = \chi_{\rho_*}(t) 
    \int_{\R^{2n}} e^{it^{\mathsf T}\xi} \, d\nu(\xi).
\end{equation}
The last integral can be computed as
\begin{equation}
\begin{split}
    \int_{\R^{2n}} e^{it^{\mathsf T}\xi} \, d\nu(\xi) &= 
    \mathbb{E}[e^{it^{\mathsf T}\xi}] = 
    e^{it^{\mathsf T}\mu} \mathbb{E}[e^{i(\sqrt{C}t)^{\mathsf T}Z}] \\
    &= e^{it^{\mathsf T}\mu} \exp\left(-\tfrac{1}{2}\|\sqrt{C}t\|^2\right) = 
    \exp\left(it^{\mathsf T}\mu - \tfrac{1}{2}t^{\mathsf T}C t\right).
\end{split}
\end{equation}

\section{Proof of Proposition~\ref{prop:robustness}}\label{app:E}

Suppose that no such tolerance existed. There would then be biseparable states
$\rho^{(\ell)}$ such that $\rho^{(\ell)}_{ij}\to\eta_{ij}$ in trace norm for
every edge. Let $M=2^{n-1}-1$ be the number of nontrivial bipartitions. Applying
Eq.~\eqref{eq:rho-bisep} to each $\rho^{(\ell)}$, at least one cut component has
weight at least $1/M$. Since there are only finitely many bipartitions, pass to
a subsequence for which one fixed bipartition $X|\overline X$ has such a
component $T_\ell$:
\begin{equation}
 0\preccurlyeq T_\ell\preccurlyeq\rho^{(\ell)},\qquad
 T_\ell\in\operatorname{Sep}_+(X|\overline X),\qquad
 \operatorname{Tr}T_\ell\geqslant1/M.
 \label{gmar:eq:robust-cut-component}
\end{equation}
Choose a fixed edge $\{i,j\}$ crossing this cut and reduce to that edge. Writing
$S_\ell=\operatorname{Tr}_{\overline{ij}}T_\ell$, we have
\begin{equation}
 0\preccurlyeq S_\ell\preccurlyeq\rho^{(\ell)}_{ij},\qquad
 \operatorname{Tr}S_\ell\geqslant1/M,\qquad
 S_\ell\in\operatorname{Sep}_+(i|j).
\end{equation}
The operators $S_\ell$ have a trace-norm convergent subsequence. To see this
without assuming convergence of unbounded moments, choose finite-rank
projections $P_L\uparrow E$ on the modes $i$ and $j$. Trace-norm convergence
$\rho^{(\ell)}_{ij}\to\eta_{ij}$ implies uniform trace tightness,
\begin{equation}
 \sup_\ell\operatorname{Tr}[(E-P_L)\rho^{(\ell)}_{ij}]
 \longrightarrow0.
 \label{gmar:eq:robust-tightness}
\end{equation}
Domination gives the same tail bound for $S_\ell$. Using the compression
estimate from the closed-hull lemma and $\operatorname{Tr}S_\ell\leqslant1$, we
obtain
\begin{equation}
 \|S_\ell-P_LS_\ell P_L\|_1
 \leqslant2\sqrt{\operatorname{Tr}S_\ell\,
                   \operatorname{Tr}[(E-P_L)S_\ell]}
 \longrightarrow0
 \label{gmar:eq:robust-compression}
\end{equation}
uniformly in $\ell$ as $L\to\infty$. For each fixed $L$, the compressed
operators form a bounded finite-dimensional set. Uniform approximation by
these sets gives relative compactness in trace norm.

Take a subsequential limit $S$. Positivity, domination, and trace-norm
continuity of the trace yield
\begin{equation}
 0\ne S\preccurlyeq\eta_{ij},\qquad \operatorname{Tr}S\geqslant1/M.
\end{equation}
Moreover, $S_\ell/\operatorname{Tr}S_\ell\to S/\operatorname{Tr}S$ in trace
norm. Closedness of the bipartite separable set makes $S$ separable across
$i|j$. Theorem~1 would then make $\eta_{ij}$ separable, a contradiction.

\bibliographystyle{unsrt}
\bibliography{references}

@misc{Leskovjanova2026,
  author = {Leskovjanov{\'a}, Olga and Baksov{\'a}, Kl{\'a}ra and
            Provazn{\'i}k, Jan and Mi{\v{s}}ta, Jr., Ladislav and Friis, Nicolai},
  title = {On the existence of fully inseparable biseparable {Gaussian} states},
  year = {2026},
  note = {arXiv:2605.28404v1},
  url = {https://arxiv.org/abs/2605.28404v1}
}

@misc{Holevo2005,
  author = {Holevo, A. S. and Shirokov, M. E. and Werner, R. F.},
  title = {Separability and Entanglement-Breaking in Infinite Dimensions},
  year = {2005},
  note = {arXiv:quant-ph/0504204v1},
  url = {https://arxiv.org/abs/quant-ph/0504204v1}
}

@article{WernerWolf2001,
  author = {Werner, R. F. and Wolf, M. M.},
  title = {Bound Entangled {Gaussian} States},
  journal = {Physical Review Letters},
  volume = {86},
  pages = {3658},
  year = {2001},
  doi = {10.1103/PhysRevLett.86.3658},
  url = {https://arxiv.org/abs/quant-ph/0009118v1}
}

@article{Giedke2001,
  author = {Giedke, G. and Kraus, B. and Lewenstein, M. and Cirac, J. I.},
  title = {Separability Criterion for All Bipartite {Gaussian} States},
  journal = {Physical Review Letters},
  volume = {87},
  pages = {167904},
  year = {2001},
  doi = {10.1103/PhysRevLett.87.167904},
  url = {https://arxiv.org/abs/quant-ph/0104050v1}
}

@article{Baksova2025,
  author = {Baksov{\'a}, Kl{\'a}ra and Leskovjanov{\'a}, Olga and
            Mi{\v{s}}ta, Jr., Ladislav and Agudelo, Elizabeth and Friis, Nicolai},
  title = {Multi-copy activation of genuine multipartite entanglement in
           continuous-variable systems},
  journal = {Quantum},
  volume = {9},
  pages = {1699},
  year = {2025},
  doi = {10.22331/q-2025-04-09-1699},
  url = {https://quantum-journal.org/papers/q-2025-04-09-1699/},
  note = {Also arXiv:2312.16570v5}
}

@article{Carvacho2017GHZGeometry,
  author  = {Gonzalo Carvacho and Francesco Graffitti and Vincenzo D'Ambrosio
             and Beatrix C. Hiesmayr and Fabio Sciarrino},
  title   = {Experimental investigation on the geometry of {GHZ} states},
  journal = {Scientific Reports},
  volume  = {7},
  number  = {1},
  pages   = {13265},
  year    = {2017},
  doi     = {10.1038/s41598-017-13124-6}
}

@article{GuhneSeevinck2010,
  author  = {Otfried G{\"u}hne and Michael Seevinck},
  title   = {Separability criteria for genuine multiparticle entanglement},
  journal = {New Journal of Physics},
  volume  = {12},
  number  = {5},
  pages   = {053002},
  year    = {2010},
  doi     = {10.1088/1367-2630/12/5/053002}
}

@article{Yang2026GaussianRigidity,
  author        = {Yu Yang and Yuhang Wang and Chunxiao Du and Shikun Zhang
                   and Zheng Qin and Rui Li and Wenxiu Li and Hao Zhang
                   and Zhisong Xiao},
  title         = {Full Inseparability and Genuine Multipartite Entanglement
                   Coincide for Finite-Mode Gaussian States},
  journal       = {arXiv preprint arXiv:2609.10984},
  year          = {2026},
  eprint        = {2609.10984},
  archivePrefix = {arXiv},
  primaryClass  = {quant-ph}
}

@article{Holevo2008EB,
  author  = {Holevo, A. S.},
  title   = {Entanglement-breaking channels in infinite dimensions},
  journal = {Problems of Information Transmission},
  volume  = {44},
  number  = {3},
  pages   = {171--184},
  year    = {2008},
  doi     = {10.1134/S0032946008030010},
  url     = {https://doi.org/10.1134/S0032946008030010}
}

@article{PhysRevA.67.062320,
  title = {Driving non-Gaussian to Gaussian states with linear optics},
  author = {Browne, Daniel E. and Eisert, Jens and Scheel, Stefan and Plenio, Martin B.},
  journal = {Phys. Rev. A},
  volume = {67},
  issue = {6},
  pages = {062320},
  numpages = {9},
  year = {2003},
  month = {Jun}
}

@article{PRXQuantum.3.010337,
  title = {No-Go Theorems for Quantum Resource Purification: New Approach and Channel Theory},
  author = {Fang, Kun and Liu, Zi-Wen},
  journal = {PRX Quantum},
  volume = {3},
  issue = {1},
  pages = {010337},
  numpages = {19},
  year = {2022},
  month = {Mar}
}

@article{Leskovjanova2025minimalcriteria,
  title = {Minimal criteria for continuous-variable genuine multipartite entanglement},
  author = {Leskovjanov{\'{a}}, Olga and Mi{\v{s}}ta Jr., Ladislav},
  journal = {{Quantum}},
  volume = {9},
  pages = {1837},
  month = aug,
  year = {2025}
}

@article{PhysRevA.106.062410,
  title = {Certifying emergent genuine multipartite entanglement with a partially blind witness},
  author = {Nordgren, Viktor and Leskovjanov{\'a}, Olga and Provazn{\'{\i}}k, Jan and Johnston, Adam and Korolkova, Natalia and Mi{\ifmmode \check{s}\else \v{s}\fi{}}ta, Ladislav},
  journal = {Phys. Rev. A},
  volume = {106},
  issue = {6},
  pages = {062410},
  numpages = {16},
  year = {2022},
  month = {Dec}
}

\end{document}